\documentclass[11pt]{article}

\usepackage{fullpage}
\usepackage{amsthm,amsmath,amssymb}
\usepackage{xcolor,xspace,hyperref}

\newsavebox{\theorembox}
\newsavebox{\factbox}
\newsavebox{\lemmabox}
\newsavebox{\corollarybox}
\newsavebox{\propositionbox}
\newsavebox{\examplebox}
\newsavebox{\conjecturebox}
\newsavebox{\algbox}
\newsavebox{\qbox}
\newsavebox{\problembox}
\newsavebox{\definitionbox}
\newsavebox{\assumptionbox}
\newsavebox{\hypothesisbox}
\newsavebox{\obsbox}
\savebox{\theorembox}{\noindent\bf Theorem}
\savebox{\factbox}{\noindent\bf Fact}
\savebox{\lemmabox}{\noindent\bf Lemma}
\savebox{\corollarybox}{\noindent\bf Corollary}
\savebox{\propositionbox}{\noindent\bf Proposition}
\savebox{\examplebox}{\noindent\bf Example}
\savebox{\conjecturebox}{\noindent\bf Conjecture}
\savebox{\algbox}{\noindent\bf Algorithm}
\savebox{\obsbox}{\noindent\bf Observation}
\savebox{\definitionbox}{\noindent\bf Definition}
\savebox{\problembox}{\noindent\bf Open Problem}
\savebox{\assumptionbox}{\noindent\bf Assumption}
\savebox{\hypothesisbox}{\noindent\bf Hypothesis}
\newtheorem{theorem}{\usebox{\theorembox}}
\newtheorem{lemma}[theorem]{\usebox{\lemmabox}}
\newtheorem{corollary}[theorem]{\usebox{\corollarybox}}
\newtheorem{proposition}[theorem]{\usebox{\propositionbox}}

\newtheorem{conjecture}{\usebox{\conjecturebox}}

\newtheorem{problem}{\usebox{\problembox}}

\newtheorem{defn}{\usebox{\definitionbox}}

\newtheorem{obs}{\usebox{\obsbox}}
\newcommand{\etal}{et al.\xspace}

\newcommand{\DNF}{\text{Monotone-DNF}}
\newcommand{\N}{\mathbb{N}}
\newcommand{\Z}{\mathbb{Z}}

\newcommand{\bx}{\mathbf{x}}
\newcommand{\by}{\mathbf{y}}

\newcommand{\bA}{\mathbf{A}}
\newcommand{\bB}{\mathbf{B}}

\newcommand{\cU}{\mathcal{U}}
\newcommand{\cS}{\mathcal{S}}
\newcommand{\cX}{\mathrm{X}}

\newcommand{\poly}{\text{poly}}
\newcommand{\cC}{\mathcal{C}}

\newcommand{\cF}{\mathcal{F}}

\newcommand{\cL}{\mathcal{L}}
\newcommand{\cH}{\mathcal{H}}

\newcommand{\cI}{\mathcal{I}}

\newcommand{\cT}{\mathcal{T}}
\newcommand{\tcT}{\widetilde{\cT}}
\newcommand{\tcS}{\widetilde{\cS}}

\newcommand{\E}{\mathbb{E}}
\DeclareMathOperator*{\EX}{\mathbb{E}}

\DeclareMathOperator*{\val}{val}
\DeclareMathOperator*{\wval}{weak-val}
\DeclareMathOperator*{\var}{var}

\DeclareMathOperator*{\maj}{Majority}

\DeclareMathOperator{\opt}{OPT}

\DeclareMathOperator{\wagr}{weak-agr}

\DeclareMathOperator{\disagr}{disagr}
\DeclareMathOperator*{\disa}{disagr}

\newcommand{\geqs}{\geqslant}
\newcommand{\leqs}{\leqslant}
\renewcommand{\geq}{\geqs}
\renewcommand{\leq}{\leqs}

\allowdisplaybreaks

\begin{document}
\title{Tight Running Time Lower Bounds for Strong Inapproximability of Maximum $k$-Coverage, Unique Set Cover and Related Problems \\ (via $t$-Wise Agreement Testing Theorem)}

\author{Pasin Manurangsi\thanks{Email: pasin@berkeley.edu. Now at Google Research.} \\ UC Berkeley}

\maketitle

\begin{abstract}
We show, assuming the (randomized) Gap Exponential Time Hypothesis (Gap-ETH), that the following tasks cannot be done in $T(k) \cdot N^{o(k)}$-time for any function $T$ where $N$ denote the input size:
\begin{itemize}
\item $\left(1 - \frac{1}{e} + \varepsilon\right)$-approximation for \textsc{Max $k$-Coverage} for any constant $\varepsilon > 0$,
\item $\left(1 + \frac{2}{e} - \varepsilon\right)$-approximation for \textsc{$k$-Median} (in general metrics) for any constant $\varepsilon > 0$.
\item $\left(1 + \frac{8}{e} - \varepsilon\right)$-approximation for \textsc{$k$-Mean} (in general metrics) for any constant $\varepsilon > 0$.
\item Any constant factor approximation for \textsc{$k$-Unique Set Cover}, \textsc{$k$-Nearest Codeword Problem} and \textsc{$k$-Closest Vector Problem}.
\item $(1 + \delta)$-approximation for \textsc{$k$-Minimum Distance Problem} and \textsc{$k$-Shortest Vector Problem} for \emph{some} $\delta > 0$. 
\end{itemize}
Since all problems considered here can be trivially solved in $N^{O(k)}$ time, our running time lower bounds are tight up to a constant factor in the exponent. In terms of approximation ratios, \textsc{Max $k$-Coverage} is well-known to admit polynomial-time $\left(1 - \frac{1}{e}\right)$-approximation algorithms, and, recently, it was shown that \textsc{$k$-Median} and \textsc{$k$-Mean} are approximable to within factors of $\left(1 + \frac{2}{e}\right)$ and $\left(1 + \frac{8}{e}\right)$ respectively in FPT time~\cite{CGKLL19}; hence, our inapproximability ratios are also tight for these three problems. For the remaining problems, no non-trivial FPT approximation algorithms are known.

The starting point of all our hardness results mentioned above is the \textsc{Label Cover} problem (with projection constraints). We show that \textsc{Label Cover} cannot be approximated to within any constant factor in $T(k) \cdot N^{o(k)}$ time, where $N$ and $k$ denote the size of the input and the number of nodes on the side with the larger alphabet respectively. With this hardness, the above results follow immediately from known reductions.

The hardness of \textsc{Label Cover} is in turn shown via a \emph{$t$-wise agreement testing theorem} of the following form: given local boolean functions $f_1, \dots, f_k$ on domains $S_1, \dots, S_k \subseteq [n]$, if random $t$ functions ``weakly agree'' with sufficiently large probability, then we can find a global boolean function $g: [n] \to \{0, 1\}$ that ``mostly agrees'' with ``many'' of the local functions. We prove such a statement in the regime where $S_1, \dots, S_k$ are ``random-looking'' sets of size $\Theta(n/k)$.
\end{abstract}

\newpage

\section{Introduction}

Since the paper of Feige \etal~\cite{FGLSS} that connects inapproximability of clique to multi-prover interactive proofs and the subsequent proof of the PCP theorem~\cite{ALMSS,AS}, the area of hardness of approximation has flourished into a fruitful avenue for research that provides (partial) answers to many classic approximation algorithm-related questions for a wide range of combinatorial optimization problems, such as constraint satisfaction problems, covering problems, and scheduling problems. An arguably surprising aspect of the field is that, despite dealing with problems of vastly different natures, most of the known reductions start from (variants of) a single problem known as the \emph{Label Cover} problem (aka \emph{projection games}), which can be defined as follows:

\begin{defn} \label{def:label-cover}
An instance $\cL = (U, V, E, \{\Sigma_u\}_{u \in U}, \{\Sigma_v\}_{v \in V}, \{\pi_e\}_{e \in E})$ of \emph{Label Cover} (aka \emph{projection games}) consists of
\begin{itemize}
\item A bipartite graph $(U, V, E)$, referred to as the \emph{constraint graph} of $\cL$,
\item For each vertex $w \in U \cup V$, an \emph{alphabet}\footnote{It should be noted that the standard definition of Label Cover often has a single alphabet set $\Sigma$ for all vertices in each side. It is not hard to see that our definition is equivalent by simply renaming the labels, but our definition is more convenient to work with in reductions.} $\Sigma_w$ for $w$, 
\item For each edge $e = (u, v) \in E$, a \emph{constraint} (aka \emph{projection}) $\pi_e: \Sigma_u \to \Sigma_v$.
\end{itemize}

A \emph{labeling} of $\cL$ is a tuple $\sigma = (\sigma_w)_{w \in U \cup V}$ where $\sigma_w \in \Sigma_w$. A constraint $\pi_{(u, v)}$ is said to be \emph{satisfied} by $\sigma$ iff $\pi_{(u, v)}(\sigma_u) = \sigma_v$. The \emph{value of a labeling} $\sigma$, denoted by $\val_{\cL}(\sigma)$, is the fraction of constraints it satisfies, i.e., $|\{(u, v) \in E \mid \pi_{(u, v)}(\sigma_u) = \sigma_v\}| / |E|$. The \emph{value of an instance} $\cL$, denoted by $\val(\cL)$, is the maximum value among all labelings, i.e., $\val(\cL) = \max_{\sigma} \val_{\cL}(\sigma)$.

In the $\nu$\textsc{-Gap-Label-Cover} problem where $\nu \in (0, 1)$, we are given an instance $\cL$ and we would like to distinguish between the following two cases: $\val(\cL) = 1$ and $\val(\cL) < \nu$. 
\end{defn}

The PCP theorem implies that $\nu$\textsc{-Gap-Label-Cover} is NP-hard for \emph{some} constant $\nu < 1$. By applying Raz's parallel repetition theorem to this, one arrives at NP-hardness of $\nu$\textsc{-Gap-Label-Cover} for \emph{any} constant $\nu > 0$~\cite{Raz98}. This has been used to prove tight NP-hardness of approximation for many problems, such as \textsc{Set Cover}~\cite{Feige98,DinurS14-parallel-rep,Moshkovitz15}, \textsc{3-SAT}~\cite{Hastad01} and other \textsc{CSP}s~\cite{Chan16}, and \textsc{Maximum Clique}~\cite{Hastad96}. In the ensuing years, variants of the \textsc{Label Cover} problem have been proposed in order to overcome certain barriers in proving inapproximability results; these include \textsc{Unique Games} and \textsc{$d$-to-1 Games}~\cite{Khot02}, \textsc{Smooth Label Cover}~\cite{Khot02-b} and \textsc{Multi-layered Label Cover}~\cite{DinurGKR05}, among others.

The past few years have seen many developments in the intersection between the areas of hardness of approximation and \emph{parameterized complexity}; the latter attempts to study computational problems on a more fine-grained level beyond the question of whether they belong to P or are NP-hard, by considering certain ``parameters'' of the problems. Specifically, in a parameterized problem, part of its input is designated as the parameter $k$, and the problem is said to be fixed parameter tractable (FPT) if it runs in time $T(k) \cdot N^{O(1)}$ where $N$ denotes the size of the input and $T$ can be any function. This FPT notion serves as the notion of efficient algorithms, instead of polynomial-time algorithms in the classic theory of NP-completeness.

Given the importance of \textsc{Label Cover} in the theory of NP-hardness of approximation, it should come as no surprise that these ``parameterized inapproximability'' results are often proved via some variants of the \textsc{Label Cover} problem. Indeed, for the case of two fundamental problems in parameterized complexity, \textsc{$k$-Clique}~\cite{CCKLM17} and \textsc{$k$-Set Cover}~\cite{KLM18}, proving hardness of approximation for the appropriate variants of \textsc{Label Cover} is the main challenge of these works; once these are proved, the inapproximability for \textsc{$k$-Clique} and \textsc{$k$-Set Cover} follow via known reductions from classic literatures on (NP-)hardness of approximation~\cite{FGLSS,Feige98}. 

Unfortunately, these variants of \textsc{Label Cover} are weaker (i.e. easier to prove inapproximability) than the original version (in Definition~\ref{def:label-cover}). In fact, the fine-grained approximability of the (original) \textsc{Label Cover} problem is not yet fully understood. For the purpose of this work, we will focus on the parameter $k = |U|$, the number of left vertices. (We discuss another parameterization in Section~\ref{sec:open}.) For this parameterization, Cohen-Addad \etal~\cite{CGKLL19} observed that the gap exponential time hypothesis (Gap-ETH)\footnote{Gap-ETH states that, for some $\varepsilon > 0$, there is no $2^{o(n)}$-time algorithm that distinguish between a satisfiable formula and one which is not even $(1 - \varepsilon)$-satisfiable. Please refer to Conjecture~\ref{conj:gap-eth} for a formal statement.} implies that there is no $T(k) \cdot N^{o(k)}$-time algorithm for $\nu$\textsc{-Gap-Label-Cover} for \emph{some} constant $\nu < 1$. 
Notice that the running time lower bound is once again tight, as there is a straightforward $N^{O(k)}$-time algorithm that works by enumerating all possible labelings to $U$.
They further observe that, by using the parallel repeition theorem~\cite{Raz98}, the gap can be amplify to \emph{any} constant $\nu > 0$. However, since the parallel repetition increases the number of left vertices from $k$ to $k^2$ in each application, the running time lower bound is now not of the form $T(k) \cdot N^{\Omega(k)}$ anymore, but rather $T(k) \cdot N^{k^{\poly(\nu)}}$. It is hence a natural question whether this running time lower bound can be improved to $T(k) \cdot N^{\Omega(k)}$.

\subsection{Our Contribution}

Our main contribution is an answer to this question: we show that, for \emph{any} constant $\nu > 0$, $\nu$-\textsc{Gap-Label-Cover} cannot be solved in $T(k) \cdot N^{o(k)}$ time assuming Gap-ETH, as stated below.

\begin{theorem} \label{thm:main-label-cover}
For any constant $\nu > 0$ and for any function $T$, assuming Gap-ETH, there is no $T(k) \cdot N^{o(k)}$-time algorithm that can solve $\nu$\textsc{-Gap-Label-Cover}, where $N$ denotes the size of the input \textsc{Label Cover} instance $\cL = (U, V, E, \{\Sigma_u\}_{u \in U}, \{\Sigma_v\}_{v \in V}, \{\pi_e\}_{e \in E})$ and $k$ denotes $|U|$.
\end{theorem}

As \textsc{Label Cover} is a starting point for many hardness of approximation reductions, our result immediately yields many corollaries. Below we point out several such consequences. We stress that all results mentioned below follow from applying known reductions to Theorem~\ref{thm:main-label-cover}, and that we by no means attempt to provide a comprehensive list of problems whose hardness follow from Theorem~\ref{thm:main-label-cover}. Indeed, given the importance of \textsc{Label Cover}, we hope that our result can serve as a starting point of many more tight running time lower bounds for inapproximability results.


Due to the fact that each of the following problems is a classic problem with many variants and long history, we will only explicitly mention works that are closely related to our results.

\paragraph{Max $k$-Coverage.}
In the \textsc{Max $k$-Coverage} problem, we are given a universe $U$, a collection $\cS$ of subsets of $U$, and a positive integer $k$. The goal is to find subsets $S_1, \dots, S_k \in \cS$ that maximizes $|S_1 \cup \cdots \cup S_k|$; an element that belongs to $S_1 \cup \cdots \cup S_k$ is said to be \emph{covered} by the $k$ subsets.

A simple greedy algorithm has long been known to provide a $(1 - 1/e)$-approximation for the problem (see e.g.~\cite{Hochba:1997:AAN:261342.571216}). The problem is known to be Max-SNP-hard~\cite{PapadimitriouY91}, and, hence, the PCP Theorem implies that it is NP-hard to approximate to within $(1 + \varepsilon)$ factor for some (small) $\varepsilon > 0$~\cite{ALMSS}. Later, Feige~\cite{Feige98} proved a tight NP-hardness of approximation with factor $(1 - 1/e + \varepsilon)$ for the problem for any $\varepsilon > 0$. This is shown via a reduction from a variant of \textsc{Label Cover} (see Theorem~\ref{thm:max-cov-feige}). Cohen{-}Addad \etal~\cite{CGKLL19} recently extends this hardness to the parameterized regime, by showing (under Gap-ETH) that there is no $(1 - 1/e + \varepsilon)$-approximation algorithm for the problem that runs in time $T(k) \cdot N^{k^{\poly(1/\varepsilon)}}$ for any $\varepsilon > 0$. This is indeed shown via a reduction from \textsc{Label Cover}, but without a tight running time lower bound. By using our Theorem~\ref{thm:main-label-cover} instead, we get the following tight running time lower bound for \textsc{Max $k$-Coverage}:

\begin{corollary} \label{cor:max-k-cov}
For any constant $\varepsilon > 0$ and any function $T$, assuming Gap-ETH, there is no $T(k) \cdot N^{o(k)}$-time algorithm that can approximate \textsc{Max $k$-Coverage} to within a factor of $(1 - 1/e + \varepsilon)$, even with a promise that there exist $k$ sets that fully cover the whole universe.
\end{corollary}

The last part regarding a promise that there exists $k$ sets that fully cover the universe is only included because this property is needed for the reduction from \textsc{Max $k$-Coverage} to the following problems we consider: \textsc{$k$-Median} and \textsc{$k$-Mean}.

\paragraph{$k$-Median and $k$-Mean.} 
In the \textsc{$k$-Median} problem (in general metric space), we are given a set $V$ of clients, a set $F$ of facilities, and a metric $d$ on $V \cup F$. The goal is to select a subset $S \subseteq F$ of $k$ facilities so that the sum of the distance of each client to his/her closest facility, i.e. $\sum_{v \in V} \min_{f \in S} d(v, f)$, is minimized. The \textsc{$k$-Mean} problem is defined similarly, except that the objective is to minimize the sum of the \emph{square} of the distance of each client to his/her closest facility, i.e. $\sum_{v \in V} \min_{f \in S} d(v, f)^2$.

There have been numerous polynomial-time approximation algorithms invented for \textsc{$k$-Median} over the years~\cite{CharikarGTS02,JainV01,JainMS02,AryaGKMMP04,LiS16,ByrkaPRST17}; the current best approximation ratio, due to Byrka \etal~\cite{ByrkaPRST17} is 2.61. For \textsc{$k$-Mean}, the problem has been more studied in the Euclidean metric, instead of the general metric as defined above. For our version of the problem, many of the aforementioned approximation algorithms for \textsc{$k$-Median} can be adapted for \textsc{$k$-Mean}, albeit with losses in the approximation ratios; these losses are often not explicitly stated in literature. There are several works that explicitly give polynomial-time algorithms for the problem and compute their approximation ratios~\cite{KanungoMNPSW04,GT08,AhmadianNSW17}, with the best one being 6.357 due to Ahmadian \etal~\cite{AhmadianNSW17}.

On the hardness front, \textsc{$k$-Median} and \textsc{$k$-Mean} are both NP-hard to approximate to within factors $(1 + 2/e - \varepsilon)$ and $(1 + 8/e - \varepsilon)$ respectively\footnote{Note that $1 + 2/e \approx 1.74$ and $1 + 8/e \approx 3.94$.} for any constant $\varepsilon > 0$~\cite{GuhaK99}. This result is shown via a reduction from \textsc{Max $k$-Coverage}. As mentioned earlier, Cohen{-}Addad \etal~\cite{CGKLL19} prove a parameterized hardness of approximating \textsc{Max $k$-Coverage}; hence, they get as a corollary parameterized hardness of approximation of factors $(1 + 2/e - \varepsilon)$ and $(1 + 8/e - \varepsilon)$ for \textsc{$k$-Median} and \textsc{$k$-Mean} respectively, but with a non-tight running time lower bound of $T(k) \cdot N^{k^{\poly(1/\varepsilon)}}$. An arguably more remarkable contribution of their work is that they give algorithms for \textsc{$k$-Median} and \textsc{$k$-Mean} that achieves approximation ratio of $(1 + 2/e + \varepsilon)$ and $(1 + 8/e + \varepsilon)$ and runs in FPT time (parameterized by $k$). Thus, their hardness of approximation is tight in terms of inapproximability ratios. Using Corollary~\ref{cor:max-k-cov}, we manage to also tighten the lower bound in terms of running time:

\begin{corollary} \label{cor:mean-median}
For any constant $\varepsilon > 0$ and any function $T$, assuming Gap-ETH, there is no $T(k) \cdot N^{o(k)}$-time algorithm that can approximate \textsc{$k$-Median} or \textsc{$k$-Mean} to within factors of $(1 + 2/e - \varepsilon)$ or $(1 + 8/e - \varepsilon)$ respectively.
\end{corollary}

\paragraph{$k$-Unique Set Cover.}
In the (exact version of) $k$-\textsc{Set Cover} problem, we are given a universe $U$ and a collection $\cS$ of subsets of $U$, and we would like to determine whether there exists $k$ subsets from $\cS$ that covers all elements in the universe $U$. In the optimization version of the $k$-\textsc{Set Cover} problem, we are given a promise that such $k$ subsets exist and we would like to find as few number of subsets from $\cS$ as possible that fully covers $U$. 

It is well known that the greedy algorithm achieves an $(\ln n + 1)$-approximation for the $k$-\textsc{Set Cover} problem~\cite{Johnson74a,Lovasz75}. On the inapproximability front, a long line of work~\cite{LundY94,Feige98,AlonMS06,DinurS14-parallel-rep,Moshkovitz15} eventually results in an NP-hardness of approximating the problem to within $(1 - \varepsilon) \ln n$ factor for any constant $\varepsilon > 0$. We remark here that this hardness is shown via (variants of) Feige's reduction~\cite{Feige98} from the \textsc{Label Cover} problem.

The parameterized approximability of \textsc{$k$-Set Cover} has long been an open question. There have been multiple recent progresses on the problem~\cite{ChenL16,CCKLM17,KLM18,Lin19}. In particular, it was shown by Karthik, Laekhanukit and the author~\cite{KLM18} that there is no FPT (in $k$) time algorithm that achieves $f(k)$-approximation for any function $f$, assuming W[1] $\ne$ FPT. Furthermore, under stronger assumptions, stronger running time lower bounds can be achieved. For instance, under the exponential time hypothesis (ETH)\footnote{ETH states that there is no $2^{o(n)}$-time algorithm for 3SAT; see Conjecture~\ref{conj:eth}.}, there is no $T(k) \cdot N^{o(k)}$-time $f(k)$-approximation algorithm for any functions $f$ and $T$. These hardness results also employ Feige's reduction; however, the starting point is not the standard \textsc{Label Cover} problem but a variant called \textsc{MinLabel}. We remark here that Lin~\cite{Lin19} recently gives an alternative proof of these results that does not seem to directly deal with any \textsc{Label Cover}-type problems.

The $k$-\textsc{Unique Set Cover} (aka $k$-\textsc{Exact Cover}) problem is a variant of the $k$-\textsc{Set Cover} problem, where we are further promised that the $k$ subsets that cover the universe only covers each element once (or equivalently that these subsets are disjoint). This variant of the problem is often considered in hardness of approximation literature, since the uniqueness of the subset covering each element can help facilitate subsequent reductions; we will see two examples of this below.

Observe that $k$-\textsc{Unique Set Cover} is easier (i.e. harder to prove hardness) than the original $k$-\textsc{Set Cover} problem. Nevertheless, Feige's reduction interestingly yields the same inapproximability for $k$-\textsc{Unique Set Cover} as well. However, the situation is quite different in the parameterized world: none of the aforementioned proof of parameterized inapproximability of the $k$-\textsc{Set Cover} problem~\cite{ChenL16,CCKLM17,KLM18,Lin19} gives any parameterized inapproximability for $k$-\textsc{Unique Set Cover}\footnote{There was a preprint~\cite{GL19} claiming to prove parameterized inapproximability for $k$-\textsc{Unique Set Cover}, but it contained a serious error and had since been withdrawn.}. A short (and slightly inaccurate) explanation of this is that the variants of \textsc{Label Cover} used in these works do not possess the \emph{projection property}, i.e., instead of having a constraint of the form $\pi_{(u, v)}: \Sigma_u \to \Sigma_v$ as in our definition, they have a constraint that can accept more than one $\sigma_v \in \Sigma_v$ for a single $\sigma_u \in \Sigma_u$. Since we now have the projection property, we overcome this and achieve hardness of approximation for $k$-\textsc{Unique Set Cover}, with a tight running time lower bound:

\begin{corollary} \label{cor:k-unique-cov}
For any function $T$, assuming Gap-ETH, there is no $T(k) \cdot N^{o(k)}$-time algorithm that can approximate \textsc{$k$-Unique Set Cover} to within any constant factor.
\end{corollary}

To the best of our knowledge, the only known parameterized hardness of approximation for the $k$-\textsc{Unique Set Cover} is from the author's (recent) dissertation~\cite{M-phd-thesis}, which states that, assuming Gap-ETH, there is no $k^{1/2 - o(1)}$-approximation algorithm for the $k$-\textsc{Unique Set Cover} problem that runs in FPT time. While the inapproximability ratio is larger than the one above, the running time lower bound there is quite poor and is only of the form $T(k) \cdot N^{k^{o(1)}}$.

\paragraph{$k$-Nearest Codeword Problem.}

In the \textsc{$k$-Nearest Codeword} problem \textsc{($k$-NCP)}, we are given a linear error-correcting code represented by its generator matrix\footnote{The set of codewords is $\{\bA\bx \mid \bx \in \mathbb{F}_2^{m}\}$.} $\bA \in \mathbb{F}_2^{n \times m}$ and a target vector $\by \in \mathbb{F}_2^n$. The goal is to find a codeword of $\bA$ that is closest to $\by$ in the Hamming distance. The parameter $k$ for this problem is the optimal distance. (This can also be thought of as a promise problem similar to \textsc{$k$-Set Cover}.)

Arora \etal~\cite{ABSS97} give a reduction from \textsc{Unique Set Cover} and \textsc{Label Cover} which shows that the problem \textsc{NCP} is as hard to approximate as those two problems. In particular, with the inapproximability of \textsc{Label Cover} from parallel repetition~\cite{Raz98}, their reduction implies that this problem is quasi-NP-hard to approximate to within $2^{\log^{1 - \varepsilon} n}$ factor.

On the parameterized inapproximability front, it was shown by Bhattacharyya \etal~\cite{BGKM18} that there is no FPT algorithm that approximates the problem to within any constant factor, assuming Gap-ETH. Later, the Gap-ETH assumption was bypassed by Bonnet \etal~\cite{BELM} who showed inapproximability of the problem under W[1]-hardness. (See also the manuscript~\cite{evenset-merged} which is a merge of~\cite{BGKM18} and~\cite{BELM}.) Nonetheless, the running time lower bounds in both cases are not yet tight: for inapproximability factor of $C$, ~\cite{BGKM18} yields a Gap-ETH-based lower bound of $T(k) \cdot N^{k^{\poly(1/C)}}$ whereas~\cite{BELM} yields an ETH-based lower bound of the form $T(k) \cdot N^{\Omega_C(\poly\log k)}$. By plugging in Arora \etal's reduction to Corollary~\ref{cor:k-unique-cov}, we immediately improve the lower bound to the tight $T(k) \cdot N^{\Omega(k)}$ for any constant factor:

\begin{corollary} \label{cor:ncp}
For any function $T$, assuming Gap-ETH, there is no $T(k) \cdot N^{o(k)}$-time algorithm that can approximate \textsc{$k$-Nearest Codeword} to within any constant factor.
\end{corollary}


\paragraph{$k$-Closest Vector Problem.}
The \textsc{$k$-Closest Vector} problem in $\ell_p$ metric (\textsc{$k$-CVP}$_p$) is similar to the \textsc{$k$-Nearest Codeword} problem, except that \textsc{$k$-CVP}$_p$ is about lattices instead of linear error-correcting codes and the distance in \textsc{$k$-CVP}$_p$ is measured in $\ell_p$ metric instead of Hamming metric. In particular, we are given a generator matrix $\bA \in \Z^{n \times m}$ of a lattice and a target vector $\by \in \Z^n$. The goal is to find a vector in the lattice that is closest in the $\ell_p$ metric to $\by$. For a technical purpose explained next, we define the objective function to be the $p$-th power of the $\ell_p$ distance, i.e., the objective is $\min_{\bx \in \Z^m} \|\bA\bx - \by\|_p^p$. This objective also serves as the parameter $k$. The reason the objective is defined this way is so that the problem is solvable in $N^{O(k)}$ time, by trying all possible $N^{O(k)}$ vectors $\by'$ with $\|\by' - \by\|_p^p \leq k$ and checking whether $\by'$ belongs to the lattice.

The approximability status of \textsc{$k$-CVP} very much mirrors that of \textsc{$k$-NCP}. In particular, Arora \etal's reduction~\cite{ABSS97} also works for \textsc{$k$-CVP}, and gives the same inapproximability ratios. The parameterized complexity of \textsc{$k$-CVP} is also similar to that of \textsc{$k$-NCP}: ~\cite{BGKM18} yields a Gap-ETH-based lower bound of $T(k) \cdot N^{k^{\poly(1/C)}}$ whereas~\cite{BELM} yields an ETH-based lower bound of the form $T(k) \cdot N^{\Omega_C(\poly\log k)}$ for approximating the problem to within a constant factor $C > 1$. Once again, by applying Arora \etal's reduction to Corollary~\ref{cor:k-unique-cov}, we immediately get the tight running time lower bound for any constant approximation algorithms for the problem:

\begin{corollary} \label{cor:nvp}
For any function $T$ and any constant $p \geq 1$, assuming Gap-ETH, there is no $T(k) \cdot N^{o(k)}$-time algorithm that can approximate \textsc{$k$-CVP$_p$} to within any constant factor.
\end{corollary}

\paragraph{$k$-Minimum Distance Problem (aka $k$-Even Set).}
The \textsc{$k$-Minimum Distance} problem (\textsc{$k$-MDP}) is the homogeneous variant of \textsc{$k$-NCP}. In \textsc{$k$-MDP}, we are given a generator matrix $\bA \in \mathbb{F}_2^{n \times m}$ of a linear error-correcting code and we would like to determine its distance, i.e., $\min_{0 \ne \bx \in \mathbb{F}_2^m} \|\bA\bx\|_0$ where $\|\cdot\|_0$ denote the Hamming weight. 

Dumer \etal~\cite{DumerMS03} were the first to prove hardness of approximation for (the non-parameterized version of) the problem. In particular, they show that, assuming NP $\ne$ RP, the problem is hard to approximate to within any constant factor; this is shown by first giving a randomized reduction from \textsc{NCP} to \textsc{MDP}, which shows \emph{some} constant factor inapproximability for the latter, and then using self-tensoring to boost the gap. Since then, the reduction has been derandomized and simplified several times~\cite{ChengW12,AustrinK14,Micciancio14}; these give NP-hardness of approximation to within any constant factor, and quasi-NP-hardness of approximation to within $2^{\log^{1 - \varepsilon} n}$ factor for the problem.

In the parameterized complexity and algorithms community, \textsc{$k$-MDP} is often referred to as the \textsc{$k$-Even Set} problem. It had been a long-standing open question whether the (exact version of) problem is FPT~\cite{DowneyFVW99}. This was recently resolved by Bhattacharyya \etal~\cite{BGKM18} who show, assuming Gap-ETH, that the problem is not only not FPT but cannot even be approximated to within any constant factor in FPT time. Similar to Dumer \etal's~\cite{DumerMS03}, their reduction is from $k$-\textsc{NCP}. As we discussed above, since they do not get a tight running time lower bound for the latter, they also only get a running time lower bound of the form $T(k) \cdot N^{k^\zeta}$ for some small $\zeta > 0$ for $k$-\textsc{MDP}. As we mentioned above, Bonnet \etal~\cite{BELM} later gave a W[1]-hardness of approximation for $k$-\textsc{NCP}, which implies also W[1]-hardness of approximating $k$-\textsc{MDP} to within any constant factor. Once again, since their running time lower bound for $k$-\textsc{NCP} is only $T(k) \cdot N^{\Omega(\poly\log k)}$, the same holds for \textsc{$k$-MDP}.

Since we manage to give a tight running time lower bound for inapproximability of $k$-\textsc{NCP} (Corollary~\ref{cor:ncp}), we can simply plug this into the reduction of~\cite{BGKM18} and obtain a tight running time lower bound for inapproximability of $k$-\textsc{MDP} for a small constant factor as well:

\begin{corollary} \label{cor:mdp}
For any function $T$ and any constant $\varepsilon > 0$, assuming Gap-ETH, no $T(k) \cdot N^{o(k)}$-time algorithm can approximate \textsc{$k$-Minimum Distance} to within a factor of $(2 - \varepsilon)$.
\end{corollary}

We remark here that the reason that we do not get tight running time lower bounds for larger factor is that the self-tensoring gap amplification, similar to parallel repetition, blows up the parameter $k$ (see the discussion in Section~\ref{sec:open}). We also note that, unlike previous problems where asymptotically tight running time lower bounds for the their \emph{exact} versions are known (or can be shown elementarily), we are not aware of any such result even for the exact version of $k$-\textsc{MDP} because the only hardness proof of the problem so far is from~\cite{BGKM18}.

\paragraph{$k$-Shortest Vector Problem.}
The $k$-\textsc{Shortest Vector} problem in $\ell_p$ norm ($k$-\textsc{SVP}$_p$) is the homogeneous variant of \textsc{$k$-CVP$_p$}. In particular, we are given a matrix $\bA \in \Z^{n \times m}$ and we want to determine $\min_{0 \ne \bx \in \Z^m} \|\bA\bx\|_p^p$. Once again, the parameter $k$ is simply the objective.

On the (non-parameterized) hardness of approximation front, Micciancio~\cite{Micciancio00} showed NP-hardness of approximating the problem to within $\sqrt[p]{2}$ factor, via a randomized reduction from \textsc{CVP}. Unlike \textsc{MDP}, gap amplification by self-tensoring for \textsc{SVP} is considerably more complicated, but it was eventually achieved by Khot~\cite{Khot05} who showed that the problem is NP-hard to approximate to within any constant factor (under randomized reductions). The gap amplification step is later simplified and extended by Haviv and Regev~\cite{HavivR12}.

The progresses on parameterized approximability of $k$-\textsc{SVP} is similar to that of $k$-\textsc{MDP}. Bhattacharyya \etal~\cite{BGKM18} observes that Khot's reduction~\cite{Khot05} from $k$-\textsc{SVP} to $k$-\textsc{CVP} also works in the parameterized setting, which implies that the latter is Gap-ETH-hard to approximate to within some constant factor in FPT time. The subsequent result of Bonnet \etal~\cite{BELM} relaxes the assumption from Gap-ETH-hardness to W[1]-hardness. Similar to above, the running time lower bounds from these two proofs are only $T(k) \cdot N^{k^c}$ for some small constant $c > 0$ (depending on $p$) and $T(k) \cdot N^{\Omega(\poly\log k)}$ respectively. Since we now have an asymptotically tight running time lower bound for inapproximability of $k$-\textsc{CVP} (Corollary~\ref{cor:nvp}), we may apply Khot's reduction~\cite{Khot05} to it, which yields:

\begin{corollary} \label{cor:svp}
For any function $T$ and any constant $p > 1$, assuming Gap-ETH, no $T(k) \cdot N^{o(k)}$-time algorithm can approximate \textsc{$k$-SVP$_p$} to within a factor of $(1 + \delta_p)$ for some $\delta_p > 0$.
\end{corollary}

\subsection{Other Related Works} \label{sec:related-work}



\subsubsection*{Agreement Testing Theorems}
As we will discuss in more details in Section~\ref{sec:main-agr-overview}, our result is proved via an \emph{agreement testing theorem}. This is a theorem of the following form: suppose that we have a collection of local functions $\{f_S\}_{S \in \cS}$ where each $f_S$ has a domain $S$ which is a subset of a universe $U$, such that it passes certain ``local agreement tests'' with sufficiently large probability. Then, one can find a global function $g$ (with domain $U$) such that $g$ ``mostly agrees'' with many of the local functions.

Theorems of this form have long been studied. For instance, many of the low degree tests (e.g.,~\cite{RazS97,AroraS03}) can be stated in this form, although they often require $f_S$ to be from a certain family of functions (e.g. low degree). On the other hand, our agreement testing theorem (see Theorems~\ref{thm:agr-informal} and~\ref{thm:agr-main}) is more combinatorial in nature and does not impose any requirements on the functions $f_S$'s. Hence, our agreement testing theorem is more closely related to the question of \emph{direct product test}~\cite{GoldreichS00,DinurR06,ImpagliazzoKW12,DinurS14,DinurN17,GoldenbergS18}, for which the collection\footnote{Sometimes the domains are viewed as tuples instead of sets, but, to the best of our knowledge, this does not significantly change the nature of the tests.} $\cS$ contains all subsets of sizes $\ell$ of the universe $U$ and the local functions can be arbitrarily. There are also combinatorial agreement theorems beyond direct product tests; for instance,~\cite{DinurK17,DiksteinD19} consider collections $\cS$ that corresponds to layered-set systems, such as high-dimensional expanders, and \cite{DinurFH19} proves a ``higher-dimensional'' version of the direct product test.

Nonetheless, all the tests mentioned in the previous paragraph consider the case where the subsets in $\cS$ are of small size compared to the universe $U$ (i.e. $\ell \ll U$), but there are a lot of sets (i.e. $|\cS| \geq |U|$). Here, we study the opposite regime of parameters: the number of our sets $|\cS|$ is small and will be the parameter $k$ of our \textsc{Label Cover} instance, whereas the size of the sets $S_1, \dots, S_k$ are very large, i.e., of the order of $\Theta(n/k)$. This is the same regime as studied in a previous work of Dinur and the author~\cite{DinurM18}. In fact, as we will discuss at the end of Section~\ref{sec:overview}, our agreement testing theorem can be seen as a generalization of that in~\cite{DinurM18}.



\subsection{Organization}

The rest of this work is organized as follows. In Section~\ref{sec:overview}, we give an overview of our proof. Next, in Section~\ref{sec:prelim}, we define additional notations that will be subsequently used. We then prove our main agreement theorem in Section~\ref{sec:main-agr}. In Section~\ref{sec:soundness-i}, we use the main agreement theorem to provide a soundness guarantee of our reduction; this guarantee is in a generic form, in the sense that the parameters are not yet specified. These parameters are calculated in Section~\ref{sec:random}, and the soundness with specific parameters are then given in Section~\ref{sec:soundness-ii}. Section~\ref{sec:lb} contains the (simple) proof of our hardness for \textsc{Label Cover}. We then explain how it implies hardness for other problems in Section~\ref{sec:consequences}. Finally, we end by posting several open questions in Section~\ref{sec:open}.

\section{Proof Overview}
\label{sec:overview}

In this section, we will describe our main reduction and the proof overview for its main properties. To do so, let us first state a few conventions, starting with additional notations for 3-SAT. A 3-CNF formula $\Phi$ consists of a set $\cX$ of variables and a set $\cC$ of clauses; each clause is an OR of at most 3 literals from $\cX$. For every clause $C \in \cC$, we use $\var(C)$ to denote the set of variables appearing in $C$. Similarly, for each subset $T \subseteq \cC$ of clauses, $\var(T)$ denote the set of variables appearing in at least one clause in $T$, i.e. $\var(T) = \bigcup_{C \in T} \var(C)$. An assignment $\phi$ of $\Phi$ is simply a function $\phi: \cX \to \{0, 1\}$. For a subset of variables $S \subseteq X$, an assignment to $S$ is a function $\phi_S: S \to \{0, 1\}$. The value of an assignment $\phi$, denoted by $\val_{\Phi}(\phi)$, is the fraction of clauses satisfied by the assignment $\phi$. When $\Phi$ is clear from the context, we may drop the subscript $\Phi$ and write $\val(\phi)$ instead of $\val_{\Phi}(\phi)$. The value of the 3-CNF formula $\Phi$ is the maximum value among all assignments, i.e. $\val(\Phi) := \max_{\phi: \cX \to \{0, 1\}} \val_{\Phi}(\phi)$. Throughout this work, we assume that each variable appears in at most $\Delta$ clauses and at least one clause; this can be assumed without loss of generality for Gap-ETH (see the discussion after Conjecture~\ref{conj:gap-eth}).

Another convention we will make is that we will use the notion of \emph{weak agreement value} instead of value for \textsc{Label Cover} in the soundness (i.e. NO case). This new notion of weak agreement value, which we will define below, is stronger (i.e. harder to prove soundness for) than the original definition of value, and hence proving soundness against weak agreement value will imply Theorem~\ref{thm:main-label-cover} as well. Historically, the notion of weak agreement value first appears in the work of Feige~\cite{Feige98} who needs to modify the parallel repeated \textsc{Label Cover} instance to achieve soundness in terms of weak agreement value. It seemed initially that achieving soundness for weak agreement value might be harder than that of the standard value; however, it was later pointed out by Moshkovitz~\cite{Moshkovitz15} that the two notions can be translated back and forth with little loss in parameters. Nonetheless, we choose to work with weak agreement value soundness since it is both more intuitive in our context and it is already compatible with Feige's reduction for \textsc{Max $k$-Coverage} (see Theorem~\ref{thm:max-cov-feige}).

We now move on to define weak agreement value. To do so, we first formalize the notion of \emph{left labeling}.
For a label cover instance $\cL = (U, V, E, \{\Sigma_u\}_{u \in U}, \{\Sigma_v\}_{v \in V}, \{\pi_e\}_{e \in E})$, its \emph{left labeling} is a tuple $\sigma = (\sigma_u)_{u \in U}$ where $\sigma_u \in \Sigma_u$. We say that a left labeling $\sigma = (\sigma_u)_{u \in U}$ \emph{weakly agrees} on a right vertex $v \in V$ iff there are two distinct neighbors $u_1, u_2$ of $v$ such that $\pi_{(u_1, v)}(\sigma_{u_1}) = \pi_{(u_2, v)}(\sigma_{u_2})$.

 The \emph{weak agreement value} of a left labeling $\sigma = (\sigma_u)_{u \in U}$, denoted by $\wval_{\cL}(\sigma)$ (or simply $\wval(\sigma)$), is defined as the fraction of right vertices on which $\sigma$ weakly agrees, i.e.,
\begin{align*}
\wval(\sigma) = \frac{|\{v \in V \mid \exists u_1 \ne u_2 \in N(v), \pi_{(u_1, v)}(\sigma_{u_1}) = \pi_{(u_2, v)}(\sigma_{u_2})\}|}{|V|}.
\end{align*}
The weak agreement value of an instance $\cL = (U, V, E, \{\Sigma_u\}_{u \in U}, \{\Sigma_v\}_{v \in V}, \{\pi_e\}_{e \in E})$ is defined as $\wval(\cL) := \max_{\sigma} \wval_{\cL}(\sigma)$.

The main theorem we will prove is the following, which is similar to Theorem~\ref{thm:main-label-cover} except that in the soundness (i.e. NO case) $\val$ is replaced by $\wval$:

\begin{theorem} \label{thm:main-label-cover-weak-agr}
Assuming Gap-ETH, the following holds. For every constant $t \in \N \setminus \{1\}$ and $\delta > 0$, and any function $T$, there is no $T(k) \cdot N^{o(k)}$-time algorithm that can, given a \textsc{Label Cover} instance $\cL = (U, V, E, \{\Sigma_u\}_{u \in U}, \{\Sigma_v\}_{v \in V}, \{\pi_e\}_{e \in E})$ of size $N$ with $k = |U|$ such that the constraint graph $(U, V, E)$ is bi-regular with right degree $t$, distinguish between the following two cases:
\begin{itemize}
\item (Completeness) $\val(\cL) = 1$.
\item (Soundness) $\wval(\cL) < \delta$.
\end{itemize}
\end{theorem} 

In Appendix~\ref{app:weak-val}, we give a simple proof of how the above theorem implies Theorem~\ref{thm:main-label-cover}. We remark here that it is needed that we can get weak agreement value soundness for \emph{any} sufficiently large integer $t$; this is because translating $\wval(\cL)$ to $\val(\cL)$ results in an additive loss of roughly $1/t$.

\subsection{The Reduction}

Having defined the required notations, we are now ready to describe our reduction from the gap version of \textsc{3-SAT} to \textsc{Label Cover}. Suppose that we start with a 3-CNF formula $\Phi = (\cX, \cC)$ with $m$ clauses. Our resulting \textsc{Label Cover} instance $\cL$ will have the following properties w.h.p.:
\begin{itemize}
\item (Completeness) If $\val(\Phi) = 1$, then $\val(\cL) = 1$.
\item (Soundness) If $\val(\Phi) < 1 - \varepsilon$, then $\wval(\cL) < \delta$.
\item (Right Degree) $\cL$ is bi-regular with right degree $t$.
\item (Size Bound) The size of $\cL$ is at most $2^{O_{\varepsilon, \delta, t}(m/k)}$ where $k$ denotes the number of right vertices in $\cL$.
\end{itemize}
Here $\varepsilon, \delta > 0$ and $t \in \N \setminus \{1\}$ can be any constants, which are taken as parameters of the reduction; notice that they effect the size of $\cL$ but, as long as $\varepsilon, \delta, t$ are constants, the size remains $2^{O(m/k)}$. 

Before we state the reduction itself, observe that, if we can give a reduction that satisfies the above properties, then we have proved Theorem~\ref{thm:main-label-cover-weak-agr}. The reason is that, if we have an $T(k) \cdot N^{o(k)}$-time algorithm that can distinguish between the two cases in Theorem~\ref{thm:main-label-cover-weak-agr}, then, given a \textsc{3-SAT} instance $\Phi$, we can run it through the above reduction and then run the algorithm on the resulting \textsc{Label Cover} instance. Since the size of the \textsc{Label Cover} instance is only $N = 2^{O(m/k)}$, the running time is $2^{o(n)}$. In other words, this algorithm can distinguish between $\val(\Phi) = 1$ and $\val(\Phi) < 1 - \varepsilon$ in $2^{o(n)}$ time, which would violate Gap-ETH.

We now move on to describe the reduction, which is quite simple and intuitive. The main idea is that, we first pick a random collection of clauses $\cT = \{T_1, \dots, T_k\}$ where each clause $C$ is included in each subset $T_i$ independently with probability $p = \frac{C}{k}$ for a sufficiently large constant $C$ (depending on $t, \varepsilon, \delta, \Delta$) which will be chosen later. Each left vertex in $U$ corresponds to a subset $T_i$; its alphabet contains all assignments to $\var(T_i)$ that satisfy all the clauses in $T_i$. Each right vertex performs an ``agreement test'' that $t$ of the assignments to the subsets agree on their intersection. More specifically, we have one right vertex for each $\{T_{i_1}, \dots, T_{i_t}\}$ with alphabet set being all assignments to $\bigcap_{j \in [t]} \var(T_{i_j})$. This right vertex has an edge to each of $T_{i_1}, \dots, T_{i_t}$ with the natural constraint: $\pi_{T_{i_j}, \{T_{i_1}, \dots, T_{i_t}\}}$ simply projects an assignment $\phi_{T_{i_j}}: T_{i_j} \to \{0, 1\}$ to its restriction $\phi_{T_{i_j}}|_{\var(\{T_{i_1}, \dots, T_{i_t}\})}$. A more formal description of the reduction is given below.

\begin{defn} \label{def:main-red}
Given a 3-CNF formula $\Phi = (\cX, \cC)$, a collection $\cT$ of subsets of $\cC$ and an integer $t \geq 2$, let the \textsc{Label Cover} instance $\cL_{\Phi, \cT, t} = (U, V, E, \{\Sigma_u\}_{u \in U}, \{\Sigma_v\}_{v \in V}, \{\pi_e\}_{e \in E})$ be as follows:
\begin{itemize}
\item The left vertex set $U$ is simply $\cT$.
\item The right vertex set $V$ is $\binom{\cT}{t}$, which contains all $t$-size subcollection of $\cT$.
\item There is an edge in $E$ between $T \in A$ and $\{T_{i_1}, \dots, T_{i_t}\} \in B$ iff $T \in \{T_{i_1}, \dots, T_{i_t}\}$.
\item The alphabet $\Sigma_T$ of $T \in A$ is the set of assignments to $\var(T)$ that satisfies all clauses in $T$.
\item The alphabet $\Sigma_{\{T_{i_1}, \dots, T_{i_t}\}}$ of $\{T_{i_1}, \dots, T_{i_t}\} \in B$ is the set of all assignments to $\bigcap_{j \in [t]} \var(T_{i_j})$.
\item The constraint $\pi_{(T, \{T_{i_1}, \dots, T_{i_t}\})}: \Sigma_T \to \Sigma_{\{T_{i_1}, \dots, T_{i_t}\}}$ sends $\phi_T \in \Sigma_T$ to its restriction $\phi_T|_{\bigcap_{j \in [t]} \var(T_{i_j})}$.
\end{itemize}
\end{defn}

Let us now check that the constructed instance satisfies the desired properties. First, it is obvious that the instance is bi-regular and that the right degree of each vertex is exactly $t$. Secondly, since each $T_i$ includes each clause independently at random with probability $p = \frac{C}{k}$, its size is $O(m/k)$ with high probability. This means that the size of the instance, which is dominated by the size of the left alphabet sets, is only $2^{O(m/k)}$ as claimed. Moreover, if $\Phi$ is satisfiable, then we may pick a left labeling of $\cL_{\Phi, \cT, t}$ to be restrictions of a satisfiying assignment of $\Phi$. Since these are restrictions of a single global assignment, all of them are consistent and $\val(\cL_{\Phi, \cT, t}) = 1$ as desired.

Hence, we are only left to prove the soundness of the reduction and this is indeed the majority of our contribution. It is more convenient to prove this contrapositively: suppose that $\wval(\cL_{\Phi, \cT, t}) \geq \delta$, i.e. there exists a left labeling $\sigma = (\sigma_T)_{T \in \cT}$ such that $\wval(\sigma_T) = 1$. We would like to ``decode'' back an assignment $\phi$ for $\Phi$ such that $\val_{\Phi}(\phi) \geq 1 - \varepsilon$.

\subsection{Rephrasing the Soundness Proof as an Agreement Testing Theorem}
\label{sec:soundness-v-agreement}

As alluded to earlier, this ``decoding'' will be done via an agreement testing theorem. To prove this, let us first formulate the notions for ``agreements'' of functions, starting with that of a pair of functions:

\begin{defn}[Disagreement of Two Functions]
For two functions $f_1: S_1 \to \{0, 1\}$ and $f_2: S_2 \to \{0, 1\}$, we write $\disagr(f_1, f_2)$ to denote $|\{u \in S_1 \cap S_2 \mid f_1(u) \ne f_2(u)\}|$. 
\end{defn}

We can then define $t$-wise weak agreement probability of a collection of functions as follows:

\begin{defn}[$t$-Wise Weak Agreement Probability]
Let $\cS$ be a collection of subsets of $[n]$. Consider a collection of functions $\cF = \{f_S\}_{S \in \cS}$ where $f_S$ is a function from $S$ to $\{0, 1\}$. The \emph{$t$-wise weak agreement probability} of $\cF$ is defined as
\begin{align*}
t\text{-}\wagr(\cF) := \Pr_{\{S_1, \dots, S_t\} \subseteq \cS}[\exists i \ne j \in [t], f_{S_i}|_{S_1 \cap \cdots \cap S_t} = f_{S_j}|_{S_1 \cap \cdots \cap S_t}].
\end{align*}
\end{defn} 

In other words, $t$-$\wagr(\cF)$ is the probability that the following agreement test accepts on collection $\cF$: pick random $t$ distinct sets $S_1, \dots, S_t$ from $\cS$ and accepts iff at least two of the functions $f_{S_1}, \dots, f_{S_t}$ agree on $S_1 \cap \cdots \cap S_t$.

We will next (informally) state our main agreement testing theorem. Recall that an agreement testing theorem is of the form: if a collection of local functions $\cF = \{f_S\}_{S \in \cS}$ passes a ``local agreement test'' with sufficiently large probability, then we can find a global function $g: [n] \to \{0, 1\}$ that ``mostly agrees'' with many of the local functions. 

Our agreement testing theorem works for the test described above (whose acceptance probability corresponds to $t$-$\wagr(\cF)$) and for the case that $\cS$ is a collection of $k$ random subsets of $[n]$ where each element is included in each subset independently with probability $\frac{C}{k}$ for a sufficiently large $C$. The theorem is stated more precisely, but still informally, below; for the formal statement, please refer to Theorem~\ref{thm:agr-main}.

\begin{theorem}[Informal Main Agreement Testing Theorem] \label{thm:agr-informal}
For any constants $\delta, t, \xi > 0$ and any sufficiently large $C, k > 0$ (depending on $\delta, t, \xi$), let $\cS$ be a collection of $k$ random subsets of $[n]$ where each element is included with probability $p = \frac{C}{k}$. Then, the following holds with probability $1 - o_{\delta, t, \xi}(1)$: for any collection $\cF = \{f_S\}_{S \in \cS}$ of functions $f_S: S \to \{0, 1\}$ such that $t$-$\wval(\cF) \geq \delta$, there exists a subcollection $\tcS \subseteq \cS$ of size $\Omega_{t, \xi}\left(\delta k\right)$ and a global function $g: [n] \to \{0, 1\}$ such that $$\E_{S \sim \tcS}[\disagr(g, f_S)] < \xi \cdot (p n).$$ 
\end{theorem}

We quickly remark here that each set $S$ is (w.h.p.) of size $O(pn)$. Hence, for small $\xi > 0$, the conclusion $\E_{S \sim \tcS}[\disagr(g, f_S)] < \xi \cdot (p n)$ states that $g$ disagrees with most of $f_S$ on a very small fraction of the domain of $f_S$.

The proof of the above theorem is indeed our main contribution, and it easily implies the soundness of our reduction, which is outlined next.
As their names suggest, $t$-$\wagr(\cF)$ is closely related to the weak agreement value of a left labeling $\sigma$ of our \textsc{Label Cover} instance $\cL_{\Phi, \cT, t}$. Specifically, we may view the left labeling $\sigma = (\sigma_T)_{T \in \cT}$ as a collection of functions $\cF = \{f_S\}_{S \in \cS}$ where $\cS = \{\var(T) \mid T \in \cT\}$ and $f_{\var(T)} = \sigma_{T}$. From this perspective, $\wval(\sigma)$ is exactly equal to $t$-$\wagr(\cF)$. Hence, if $\wval(\sigma) \geq \delta$, then Theorem~\ref{thm:agr-informal} gives us a global assignment $\phi: \cX \to \{0, 1\}$ and an $\Omega(k)$-size subcollection $\tcT \subseteq \cT$ such that $\E_{T \sim \tcT}[\disagr(\phi, \sigma_T)] < \xi \cdot (pn)$, i.e. $\phi$ ``mostly agrees'' with many local functions. Intuitively, this assignment should satisfy most of the clauses because each local function $\sigma_T$ satisfies all clauses in $T$. This is indeed true: for sufficiently large $C$ and sufficiently small $\xi$, the condition $\E_{T \sim \tcT}[\disagr(\phi, \sigma_T)] < \xi \cdot (pn)$ implies that $\val_{\Phi}(\phi) \geq 1 - \varepsilon$, as desired. (See Lemma~\ref{lem:decode-good-assignment} for the formalization of this.)

The readers might have noticed an inaccuracy in the above argument: while $\cT = \{T_1, \dots, T_k\}$ consist of random subsets of clauses, $\cS = \{\var(T_1), \dots, \var(T_k)\}$ is clearly \emph{not} distributed as required in our agreement testing theorem. In particular, two variables in a clause are more likely to appear together in $\var(T_j)$ than those that do not share a clause. While this does make our proof somewhat more technical, it turns out that this does not pose a significant issue, because our argument for the agreement testing theorem actually works for any $\cS$ provided that they are ``sufficiently random-looking''. The exact requirements for $\cS$ can be found in the formal statement of the theorem (Theorem~\ref{thm:agr-main}).

\subsection{Proving the Agreement Testing Theorem}
\label{sec:main-agr-overview}

We have now reduced our task to simply proving the agreement testing theorem (Theorem~\ref{thm:agr-informal}). The general structure of the proof of Theorem~\ref{thm:agr-informal} is similar to previous works, especially that of Dinur and the author~\cite{DinurM18}.

A general starting step of known proofs of agreement testing theorems is to define an appropriate notion of multi-level consistencies between functions. The point here is that, while the (small) acceptance probability only implies that a small fraction of pairs of functions are ``strongly consistent'', there will be structural properties between different level of consistencies that we can exploit so that we can ``zoom in'' to a smaller collection of functions such that most pairs are ``weakly consistent''. Once this is established, one typically lets the desired global function $g$ be simply the majority of these functions. To apply such a framework to our context, we mainly have to specify three things:
\begin{itemize}
\item The notions of ``strong consistency'' and ``weak consistency''.
\item The structural properties relating ``strong consistency'' and ``weak consistency'', and how it allows us to ``zoom in'' to a smaller collection of functions that are mostly ``weakly consistent''.
\item Prove that taking the majority of functions in the collection from the previous step yields a desired global function $g$. 
\end{itemize}

\paragraph{Strong and Weak Consistencies.}
As stated earlier, we would like to define the ``strong consistency'' between pairs of functions so that it can be easily related to the acceptance probability of the test, which in our case is $t$-$\wval(\cF) \geq \delta$. In this case, we define it as follows: two functions $f_{S_1}, f_{S_2} \in \cF$ are ``strongly consistent'' if, for a random subsets $\{S_{i_1}, \dots, S_{i_{t - 2}}\} \subseteq \cS \setminus \{S_1, S_2\}$, we have $f_{S_1}|_{S_1 \cap S_2 \cap S_{i_1} \cap \cdots \cap S_{i_{t - 2}}} = f_{S_2}|_{S_1 \cap S_2 \cap S_{i_1} \cap \cdots \cap S_{i_{t - 2}}}$ with probability at least $\beta$. It is not hard to see that, when we set $\beta$ to be a sufficiently small constant (depending on $\delta, t$), at least $\Omega_t(\delta)$ fraction of pairs of functions in $\cF$ must be strongly consistent. 

The notion of ``weak consistency'' of a pair of functions is a bit less intuitive to define, but will make more sense when we describe the structural property we want in the next paragraph. In particular, for a sufficiently small constant $\alpha > 0$, two functions $f_{S_1}, f_{S_2} \in \cF$ are ``weakly consistent'' if, for a random subsets $\{S_{i_1}, \dots, S_{i_{2t - 3}}\} \subseteq \cS \setminus \{S_1, S_2\}$, we have $f_{S_1}|_{S_1 \cap S_2 \cap S_{i_1} \cap \cdots \cap S_{i_{2t - 3}}} = f_{S_2}|_{S_1 \cap S_2 \cap S_{i_1} \cap \cdots \cap S_{i_{2t - 3}}}$ with probability at least $\alpha$. Similar to $\beta$, we will pick $\alpha$ to be a very small constant (much smaller than $\beta$); note that for a sufficiently small choice of $\alpha$, it will also be the case that ``strong consistency'' implies ``weak consistency''.

\paragraph{The Structural Property: Red-Blue Transitivity.}
The structural property we will use is the notion of red-blue transitivity, as defined in~\cite{DinurM18}. A red-blue graph is a graph where each edge is colored either red or blue:
\begin{defn}[Red-Blue Graphs]
A \emph{red-blue graph} is an undirected graph $G = (V, E = E_r \cup E_b)$ where the edge set $E$ is partition into two sets $E_r$, the set of red edges, and $E_b$, the set of blue edges.

We use prefix ``red-'' and ``blue-'' to refer to the corresponding notion in the red graph $(V, E_r)$ and the blue graph $(V, E_b)$ respectively. For instance, $u$ is said to be a red-neighbor of $v$ if $\{u, v\} \in E_r$.
\end{defn}

We let the \emph{two-level consistency graph} be a red-blue graph where the vertex set is the set of functions in $\cF$, the blue edges represent strong consistency, and non-edges represents weak consistency. (That is, a red edge means that the endpoints are \emph{not} even weakly consistent.) The key property this graph has is the so-called red-blue transitivity, defined below:

\begin{defn}[Red-Blue Transitivity]
For any positive integer $h \in \mathbb{N}$, a red-blue graph $G = (V, E = E_r \cup E_b)$ is said to be $h$-red-blue-transitive if, for every red edge $\{u, v\} \in E_r$, $u$ and $v$ have less than $h$ common blue-neighbors.
\end{defn}

With an appropriate parameter $h$ (i.e. $h = \Theta(\alpha/\beta^2) \cdot k$), it is not hard to show that the two-level consistency graph is $h$-red-blue-transitive. Specifically, consider any pair $\{f_{S_1}, f_{S_2}\}$ that contains $h$ common blue-neighbors $f_{S_{j_1}}, \dots, f_{S_{j_h}}$. For each common neighbor $S$, since $\{S_1, S\}$ and $\{S_2, S\}$ are both blue edges, there must be many collections $\{S_{i_1}, \dots, S_{i_{t - 2}}\}$ on which $f_{S_1}$ and $f_S$ are consistent and many collections $\{S'_{i_1}, \dots, S'_{i_{t - 2}}\}$ on which $f_S$ and $f_{S_2}$ are consistent. This means that $f_{S_1}$ and $f_{S_2}$ are consistent on $\{S\} \cup \{S_{i_1}, \dots, S_{i_{t - 2}}\} \cup \{S'_{i_1}, \dots, S'_{i_{t - 2}}\}$, which is of size at most $(2t - 3)$. Since there are many choices for each of the three parts of the union, it must be that $f_{S_1}$ and $f_{S_2}$ agree on many $(2t - 3)$-size subcollections; for an appropriately chosen $h$, this implies that $f_{S_1}$ and $f_{S_2}$ are weakly consistent. (The formal argument is given in Section~\ref{sec:rb-transitive}.)

Red-blue-transitivity is useful because, as shown in~\cite{DinurM18} (see Lemma~\ref{lem:rb-subgraph}), when the graph is $h$-red-blue-transitive for a sufficiently small $h$ and the number of blue edges is sufficiently large, it is always possible to find a large subgraph that contains almost no red edges. In our context, this is a subcollection $\tcS \subseteq \cS$ such that most pairs of functions in $\{f_S\}_{S \in \tcS}$ are weakly consistent.

\paragraph{Majority Decoding a Global Function.}
Now that we have found a large collection of functions such that most pairs are weakly consistent. We now take the global function $g$ to be the majority of the functions in the collection, and we need to argue that $g$ ``mostly agrees'' with most of the functions in the collection. 

It turns out that, for this step to go through, one only needs that the collection $\tcS$ is ``sufficiently uniform'' and that, for every pair $f_{S_1}$ and $f_{S_2}$ of weakly consistent functions, $\disagr(f_{S_1}, f_{S_2})$ is small. (See Lemma~\ref{lem:majority}.) For random collections, the former holds straightforwardly. Hence, we are only left to show the latter. To get an intuition behind why this should hold, recall that, for weakly consistent $f_{S_1}$ and $f_{S_2}$, they must be consistent on many $(2t - 3)$-size subcollections $\cS_1, \dots, \cS_{\Theta_t(\beta k^{2t - 3})} \subseteq \cS$. For each such subcollection $\cS_i$, we have $f_{S_1}$ and $f_{S_2}$ agree on $S_1 \cap S_2 \cap \left(\bigcap_{S \in \cS_i} S\right)$. Hence, by taking the union over all $i$, the two functions must agree on $S_1 \cap S_2 \cap \left(\bigcup_{i} \bigcap_{S \in \cS_i} S\right)$. The key point here is that, for random subsets, if $\beta$ is sufficiently large, then this already covers almost all the universe, and hence $f_{S_1}, f_{S_2}$ can disagree on only few coordinates. This part of the proof is formalized through the notion of \emph{strong intersection disperser} (Definition~\ref{def:intersection-disperser}) and the parameters are calculated in Section~\ref{sec:strong-int-disp}.

This concludes our proof overview.
We note here that, in our actual proof in Section~\ref{sec:main-agr}, we do not define ``strong consistency'' and ``weak consistency'' but define the two-level consistency graph directly. Nonetheless, the outline above remains the same as in the actual proof.

\paragraph{Remark on Similarities and Differences from Previous Works.} 
As stated above, the general outline of our proof is similar to that of~\cite{DinurM18}; indeed, even some structural lemmas (such as a lemma for finding a large almost non-red subgraph) from~\cite{DinurM18} are used in black-box manners in our proof. Nevertheless, the details of the two proofs are still quite different and we would like to stress a couple of points here:
\begin{itemize}
\item \textbf{2-Wise vs $t$-Wise Tests.} The main result of~\cite{DinurM18} is an agreement testing theorem, which is almost the same as ours but only works when $t = 2$. That is, the test there is simply just: pick two random functions $f_{S_1}$ and $f_{S_2}$, and accept if they agree on the intersection of their domains. In this sense, our work can be viewed as a generalization of~\cite{DinurM18}.

For $t = 2$ (in~\cite{DinurM18}), the notion of strong consistency is simply that the two functions agree on their intersection. This simplifies many steps. For instance, to see that the weakly consistent pairs disagree on few coordinates, one only needs to show that union of sufficiently many random subsets cover almost the whole universe (i.e. it is a disperser), which can be shown via a single Chernoff bound. On the other hand, we have to formulate the aforementioned notation of strong intersection disperser, whose proof is non-trivial (Section~\ref{sec:random-int-disperser}).
\item \textbf{Tightness of Parameters.} Even for the case $t = 2$, the parameters achieved in~\cite{DinurM18} is far from tight. In fact, the main agreement theorem in~\cite{DinurM18} does not\footnote{Specifically, the main agreement theorem (Theorem 21) in~\cite{DinurM18-arxiv} is only non-trivial when $\beta n < |S|$, but $\beta$ is always at least $1/\sqrt{k}$. Note here that the conference version~\cite{DinurM18} has an even poorer quantitative parameters and can only give a running time lower bound of the form $T(k) \cdot N^{\Omega(\log k)}$.} give anything non-trivial for the case where the sets $S_1, \dots, S_k$ are of size at most $O(\frac{n}{\sqrt{k}})$; equivalently, the quantitative parameters there do not give a lower bound beyond $T(k) \cdot N^{\Omega(\sqrt{k})}$ for \textsc{Label Cover}. To get the asymtotically tight lower bounds, we have to be more careful in our analysis; for instance, our agreement analysis needs to use the fact that the restriction of the collection onto $S_i \cap S_j$ are ``random-like'' with sufficiently good parameters, which is not needed in~\cite{DinurM18}. (See the requirements in Theorem~\ref{thm:agr-main}.) Such subtleties are required to make sure that the parameters are asymtotically tight.
\end{itemize}
On the other hand,~\cite{DinurM18} does face some challenges that are not present in this paper, because they would like their test to work even for a very small agreement $\delta \approx \frac{1}{k^{1 - o(1)}}$. To achieve this goal,~\cite{DinurM18} has to consider a generalized notion of red-blue transitivity. In contrast, we do not have to do so here since we only focus on the case where $\delta$ is an arbitrarily small positive constant.

\section{Preliminaries}
\label{sec:prelim}

\subsection{``Random-Like'' Properties of Set Systems}

A \emph{set system} is a pair $(U, \cS)$ where $U$ is the universe (i.e. ground set) and $\cS$ is a collection of subsets of $U$. When it is clear from the context what $U$ is, we might refer to $\cS$ itself as a set system. For a set $A \subseteq U$, we use $\cS|_A$ to denote $\{S \cap A \mid S \in \cS\}$; we call $\cS|_A$ (resp. $(A, \cS|_A)$) the restriction of $\cS$ (resp. $(U, \cS)$) onto $A$.

As stated earlier, we have to prove an agreement testing theorem that works not only with just actually random subsets but also ``random-like'' subsets, because we will apply it on the sets of variables appearing in random subsets of clauses. Towards this goal, we define two properties of ``random-like'' subsets that are required for the agreement testing theorem; we will show in Section~\ref{sec:random} that even the subsets of variables appearing in random subsets of clauses satisfy these properties with strong parameters.

\subsubsection{Uniformity}

The first property is \emph{uniformity}, which was also used in~\cite{DinurM18}. Informally, uniformity says that most of the elements appear in ``many subsets'' in the collection $\tcS$.

\begin{defn}[Uniformity] \label{def:uniform-subset}
A set system $(U, \tcS)$ is $(\gamma, \mu)$-uniform if, for at least $(1 - \mu)$ fraction of elements $u \in \cU$, $u$ appears in at least $\gamma$ fraction of the subsets in $\tcS$. In other words, $\tcS$ is $(\gamma, \mu)$-uniform if and only if $|\{u \in U \mid \Pr_{S \sim \tcS}[u \in S] \geqs \gamma\}| \geqs (1 - \mu)|U|$.
\end{defn}

\subsubsection{Strong Intersection Dispersers}

The second property is what we call \emph{strong intersection dispersers}. To motivate the definition, recall the notion of \emph{disperser}; roughly speaking, a set system $(U, \cS)$ is a disperser if, for any sufficiently large number of subsets from $\cS$, their union covers almost the whole universe. The difference in strong intersection disperser is that we replace their union with their \emph{union of intersections}. This is stated more formally below.

\begin{defn}[Strong Intersection Disperser] \label{def:intersection-disperser}
A set system $(U, \cS)$ is an \emph{$(r, \ell, \eta)$-strong intersection disperser} if, for any $r$ distinct subcollections $\cS_1, \dots, \cS_r \subseteq \cS$ each of size at most $\ell$, we have
\begin{align} \label{eq:int-disperser}
\left|U \setminus \left(\bigcup_{j \in [r]}\left(\bigcap_{S \in \cS_j} S\right)\right)\right| \leqs \eta \cdot |U|.
\end{align}
\end{defn}

We note that the ``strong'' in the name of strong intersection disperser is added because, in~\cite{DinurM18}, ``intersection disperser'' was used for a similar definition but with a condition that $\cS_1, \dots, \cS_r$ are disjoint. Indeed, it was pretty easy to see that random subsets form a (not strong) intersection disperser with certain parameters, because the indicator variable for whether an element is in $\bigcap_{S \in \cS_j} S$ is independent for each $j \in [r]$. On the other hand, this is more challenging for strong intersection dispersers that we are using here. Indeed, a whole subsection (Section~\ref{sec:random-int-disperser}) is devoted to proving such a statement.

\subsection{Exponential Time Hypotheses}

Recall that the exponential time hypothesis (ETH) states that there is no $2^{o(m)}$-time algorithm that can solve 3-SAT where $m$ denote the number of clauses. Note that the conjecture remains equivalent even when we replace $m$ with $n$, the number of variables. Specifically, the sparsification lemma~\cite{ImpagliazzoPZ01} implies that we may assume without loss of generality that each variable appears in a bounded number of clauses. Hence, ETH may be stated as follows.

\begin{conjecture}[(Randomized) Exponential Time Hypothesis (ETH)~\cite{ImpagliazzoP01,ImpagliazzoPZ01}] \label{conj:eth}
There exist constants $\kappa, \Delta > 0$ such that no randomized algorithm can, given a 3-CNF formula $\Phi$ in which each variable appears in at most $\Delta$ clauses, decides whether $\Phi$ is satisfiable (correctly with probability 2/3) in time $O(2^{\kappa m})$ where $m$ denote the number of clauses.
\end{conjecture}

The gap exponential time hypothesis (Gap-ETH) is a more recent conjecture, which is a strengthening of ETH. It states that even the gap version of 3-SAT, where we have to distinguish between $\val(\Phi) = 1$ and $\val(\Phi) < 1 - \varepsilon$ for some constant $\varepsilon > 0$, cannot be solved in $2^{o(m)}$ time. Similar to above, we may assume without loss of generality that each variable appears in a bounded number of clauses. The statement of Gap-ETH is stated more formally below.

\begin{conjecture}[(Randomized) Gap Exponential Time Hypothesis (Gap-ETH)~\cite{D16,MR17-icalp}] \label{conj:gap-eth}
There exist constants $\varepsilon, \kappa, \Delta > 0$ such that no randomized algorithm can, given a 3-CNF formula $\Phi$ in which each variable appears in at most $\Delta$ clauses, distinguish between the following two cases (correctly with probability 2/3) in time $O(2^{\kappa m})$ where $m$ denote the number of clauses:
\begin{itemize}
\item (Completeness) $\val(\Phi) = 1$.
\item (Soundness) $\val(\Phi) < 1 - \varepsilon$.
\end{itemize}
\end{conjecture}

\section{Main Agreement Theorem}
\label{sec:main-agr}

In this section, we will prove our main agreement theorem, for any set system that satisfies ``random-like'' properties with certain parameters. The theorem can be stated as follows.

\begin{theorem} \label{thm:agr-main}
For any $0 < \eta, \rho, \alpha < 1$ and $t, k, n, d \in \N$ such that $t \geqs 2$ and $k \geq 10t/\alpha$, let $\cS$ be any collection of $k$ subsets of $[n]$ such that, for every $i, j \in [k]$, $|S_i \cap S_j| \leq \rho n$ and $(\cS \setminus \{S_i, S_j\})|_{S_i \cap S_j}$ is an $(\frac{\alpha}{(10t)^{2t}} \cdot k^{2t - 3}, 2t - 3, \eta)$-strong intersection disperser.

Let $\cF = \{f_S\}_{S \in \cS}$ be any collection of functions, and let $\delta := t$-$\wagr(\cF)$. Then, there exists a subcollection $\cS' \subseteq \cS$ of size at least $\frac{\delta k}{8 t^2}$ and a function $g: [n] \to \{0, 1\}$ such that $$\EX_{S \in \cS'}\left[\disa(g, f_S)\right] \leqs n \cdot \sqrt{\rho} \cdot \sqrt{\frac{2048 t^8\alpha}{\delta^4} + \eta}.$$
\end{theorem}

Since there are many parameters in Theorem~\ref{thm:agr-main}, let us point out what they are, when each subset in $\cS$ is a random subset that contains each element independently with probability $p = \frac{C}{k}$. In this case, we have $\rho \approx p^2$. Moreover, $\alpha = \alpha(C)$ and $\eta = \eta(C)$ decreases and approaches zero as $C \to \infty$. Thus, by picking $C$ sufficiently large, we can get a guarantee of the form $\E_{S \sim \cS'}[\disagr(g, f_S)] < \xi \cdot (pn)$ for any $\xi > 0$. In other words, $g$ disagrees with $f_S$ only on $\xi$ fraction of the domain of $f_S$ on average. A guarantee of this form will be suffice for our application to the soundness analysis of our reduction.

We now proceed to the proof of the theorem; its structure is as outlined in Section~\ref{sec:main-agr-overview}.

\subsection{Two-Level Consistency Graph}

We start by defining our two-level consistency graph, which will be the same as described in Section~\ref{sec:main-agr-overview}. For notational convenient, let us define a few more notations. We say that two functions $f_{S_1}, f_{S_2} \in \cF$ are \emph{consistent on} $\cS' \subseteq \cS$ iff $f_{S_1}|_{S_1 \cap S_2 \cap (\bigcap_{S \in \cS'} S)} = f_{S_2}|_{S_1 \cap S_2 \cap (\bigcap_{S \in \cS'} S)}$. For a parameter $\ell \in \N$ and $\chi \in [0, 1]$, we say that two functions $f_{S_1}, f_{S_2} \in \cF$ are $(\ell, \chi)$-consistent iff, when we pick an $\ell$-size subcollection $\cS'$ of $\cS \setminus \{S_1, S_2\}$ uniformly at random, the probability that $f_{S_1}$ and $f_{S_2}$ are consistent on $\cS'$ is at least $\chi$. For $\ell = 0$, we say that $f_{S_1}$ and $f_{S_2}$ are $(0, \chi)$-consistent for all $\chi \in [0, 1]$.

With these notations, the two-level consistency graph can be defined as follows.

\begin{defn}[Two-Level Agreement Graph]
Given a collection of functions $\cF = \{f_S\}_{S \in \cS}$, two real numbers $0  \leq \alpha \leq \beta \leq 1$ and a positive integer $t \geq 2$, we define the two-level consistency graph $G^{\cF, \alpha, \beta, t} = (V^{\cF, \alpha, \beta, t}, E_r^{\cF, \alpha, \beta, t}, E_b^{\cF, \alpha, \beta, t})$ as follows:
\begin{itemize}
\item The vertex set $V^{\cF, \alpha, \beta, t}$ is the collection $\cS$.
\item The blue edges are the pairs $\{S_1, S_2\}$ that are $(t - 2, \beta)$-consistent.
\item The red edges are the pairs $\{S_1, S_2\}$ that are \emph{not} $(2t - 3, \alpha)$-consistent.
\end{itemize}
\end{defn}

\subsubsection{Red/Blue-Transitivity of Two-Level Consistency Graph}
\label{sec:rb-transitive}

As stated in the proof overview, the main structural property needed for the two-level consistency graph is that it is $h$-red-blue-transitive for an appropriate value of $h$. The quantitative value of $h$ is stated and proved below.

\begin{obs} \label{obs:rb-transitive}
For any $0 < \alpha \leq \beta \leq 1$ and any $t, k \in \N$ such that $t \geq 2$ and $k \geq 10t/\alpha$, the graph $G^{\cF, \alpha, \beta, t}$ is $\left(\frac{2\alpha}{\beta^2} \cdot k\right)$-red-blue-transitive.
\end{obs}

\begin{proof}
Consider any pair of distinct subsets $S_1, S_2 \in \cS$, and suppose that they have $q \geq (2 \alpha/\beta^2)k$ common blue-neighbors $S_{i_1}, \dots, S_{i_q}$. For each such common neighbor $S_{i_j}$, the definition of blue edges implies that there exist $(r - 2)$-size subcollections $\cS_{1}, \dots, \cS_{\beta \cdot \binom{k}{t - 2}}$ on which $f_{S_{1}}$ and $f_{S_{i_j}}$ are consistent and $\cS'_1, \dots, \cS'_{\beta \cdot \binom{k}{t - 2}}$ on which $f_{S_2}$ and $f_{S_{i_j}}$ are consistent. This means that, for any $u, v \in [\beta \cdot \binom{k}{t - 2}]$, we have $f_{S_1}$ and $f_{S_2}$ are consistent on $\cS_u \cup \cS'_v \cup \{S_{i_j}\}$.

Now, for every collection $\cS^* \subseteq \cS$ of size at most $(2t - 3)$, the number of collections $\cS_u, \cS'_v, \{S_{i_j}\}$ such that $\cS_u \cup \cS'_v \cup \{S_{i_j}\} = \cS^*$ is at most $(2t - 3) \cdot \binom{2t - 4}{t - 2}$. As a result, $f_{S_1}$ and $f_{S_2}$ are consistent at least
\begin{align*}
\frac{\left(\beta \cdot \binom{k}{t - 2}\right)^2 \cdot q}{(2t - 3) \cdot \binom{2t - 4}{t - 2}} \geq 2\alpha \cdot \binom{k}{2t - 3}.
\end{align*}
subcollections of $\cS \setminus \{S_1, S_2\}$ of size at most $(2t - 3)$. Since there are only $\binom{k}{2t - 4} + \cdots + \binom{k}{0}$ subcollections of $\cS \setminus \{S_1, S_2\}$ of size at most $(2t - 4)$, the number of $(2t-3)$-size subcollections of $\cS \setminus \{S_1, S_2\}$ on which $f_{S_1}$ and $f_{S_2}$ are consistent on is at least
\begin{align*}
2\alpha \cdot \binom{k}{2t - 3} - \left(\binom{k}{2t - 4} + \cdots + \binom{k}{0}\right)
\geq \alpha \cdot \binom{k}{2t - 3},
\end{align*}
where the inequality follows from $k \geq 10t/\alpha$.
In other words, $f_{S_1}$ and $f_{S_2}$ are $(2t - 3, \alpha)$-consistent, which means that $\{S_1, S_2\}$ is not a red edge, as desired.
\end{proof}

We will use the following lemma from~\cite{DinurM18} which will allow us to find a large almost non-red subgraph from the two-level consistency graph.

\begin{lemma}[{\cite[Lemma26]{DinurM18-arxiv}}] \label{lem:rb-subgraph}
For every $k, q, d \in \mathbb{N}$ and every $k$-vertex $q$-red-blue-transitive graph $G = (V, E_r \cup E_b)$ such that $|E_b| \geq 2k d$, there exists a subset $B \subseteq V$ of size at least $d$ such that $|\{(u, v) \in B \times B \mid \{u, v\} \notin E_r\}| \geq |U|^2(1 - \frac{q k}{d^2})$.
\end{lemma}

\subsubsection{Bounding Disagreements on Non-Red Edges}

The last component we need is a bound on disagreements between the endpoints of non-red edges, which follows immediately from the fact that $\cS$ is an $(\frac{\alpha}{(10t)^{2t}} \cdot k^{2t - 3}, 2t - 3, \eta)$-strong intersection disperser, as stated more formally below.

\begin{obs} \label{obs:non-red-consistent}
For any $\{S_1, S_2\} \notin E^{\cF, \alpha, \beta, t}_r$, we have $\disagr(f_{S_1}, f_{S_2}) \leq \rho \eta n$.
\end{obs}

\begin{proof}
Since $\{S_1, S_2\} \notin E^{\cF, \alpha, \beta, t}_r$, there exist $(2t - 3)$-size subcollections $\cS_1, \dots, \cS_{\alpha \cdot \binom{k}{2t - 3}} \subseteq \cS \setminus \{S_1, S_2\}$ on which $f_{S_1}$ and $f_{S_2}$ are consistent. 

For each $j \in [\alpha \cdot \binom{k}{2t - 3}]$, suppose that $\cS_j = \{S_{i_1}, \dots, S_{i_{2r - 3}}\}$; we may define $\cS'_j \subseteq (\cS \setminus \{S_1, S_2\})|_{S_1 \cap S_2}$ as $\{S_{i_1} \cap (S_1 \cap S_2), \dots, S_{i_{2r - 3}} \cap (S_1 \cap S_2)\}$. Notice that $f_{S_1}$ and $f_{S_2}$ agree on $\cS'_1, \dots, \cS'_{\alpha \cdot \binom{k}{2t - 3}}$ with respect to the set system $(\cS \setminus \{S_1, S_2\})|_{S_1 \cap S_2}$. From our assumption, this set system is an $(\frac{\alpha}{(10t)^{2t}} \cdot k^{2t - 3}, 2r - 3, \eta)$-strong intersection disperser. From this and from $\alpha \cdot \binom{k}{2t - 3} \geq \frac{\alpha}{(10t)^{2t}} \cdot k^{2t - 3}$, we have $\disagr(f_{S_1}, f_{S_2}) \leq \eta \cdot |S_1 \cap S_2| \leq \rho \eta n$ as desired.
\end{proof}

The above lemma means that the non-red pairs disagree on only a small fraction of coordinate. This in turn is sufficient for majority decoding to give a global function that mostly agrees with most of the local functions from the almost non-red subgraph. The formalization of this is stated below; we note here that the lemma is essentially the same as Lemma 27 in~\cite{DinurM18-arxiv}, except that we have a slightly sharpened bound that takes into account the bound on $|S_i \cap S_j|$. For completeness, we give the proof of this lemma is Appendix~\ref{app:majority}.

\begin{lemma}[\cite{DinurM18-arxiv}] \label{lem:majority}
For any $\zeta, \kappa, \rho > 0$, if $\cF = \{f_{S}\}_{S \in \cS'}$ is a collection of functions such that $\Pr_{S_1, S_2 \in \cS'}[\disagr(f_{S_1}, f_{S_2}) \leq \zeta \cdot n] \geq 1 - \kappa$ and that $|S_1 \cap S_2| \leq \rho n$ for all distinct $S_1, S_2 \in \cS'$, then the function $g: [n] \to \{0, 1\}$ defined by $g(x) := \maj_{S \in \cS' \atop S \ni x} f_S(x)$ satisfies
\begin{align*}
\E_{S \in \cS'}[\disagr(g, f_S)] \leq n\sqrt{\rho\kappa + \zeta}.
\end{align*}
\end{lemma}

\subsection{Putting Things Together}

Now that we have all the pieces ready, we can plug them together and prove Theorem~\ref{thm:agr-main}.

\begin{proof}[Proof of Theorem~\ref{thm:agr-main}]
Since $r$-$\wagr(\cF) = \delta$, there exists at least $\delta \binom{k}{t}$ $t$-size subcollections $\{S_1, \dots, S_t\}$ such that, for some $i \ne j \in [t]$, $f_{S_i}$ and $f_{S_j}$ are consistent on $\{S_1, \dots, S_r\}$; this latter condition is equivalent to $f_{S_i}$ and $f_{S_j}$ are consistent on $\cS' = \{S_1, \dots, S_{i - 1}, S_{i + 1}, \dots, S_{j - 1}, S_{j + 1}, \dots, S_r\}$. Consider the map $\{f_{S_1}, \dots, f_{S_r}\} \mapsto (\{S_i, S_j\}, \cS')$. Observe that this is an injection; as a result, there are at least $\delta \binom{k}{t}$ distinct $(\{S_i, S_j\}, \cS')$'s such that $f_{S_i}$ and $f_{S_j}$ are consistent on $(r - 2)$-size subcollection $\cS'$. Since $\delta \binom{k}{t} > \binom{k}{2} \cdot \left(\frac{\delta}{4t^2}\right)\binom{k}{t - 2} + \left(\frac{\delta}{4t^2}\right)k^2 \cdot \binom{k}{t - 2}$, there are at least $\left(\frac{\delta}{4t^2}\right)k^2$ pairs $\{S_i, S_j\}$ such that $S_i$ and $S_j$ agree on at least $\left(\frac{\delta}{4t^2}\right)\binom{k}{t - 2}$ $(t - 2)$-size subcollections (i.e. $S_i$ and $S_j$ are $(t - 2, \frac{\delta}{4t^2})$-consistent).

Now, let $\beta = \frac{\delta}{4t^2}$ and consider the two-level agreement graph $G^{\cF, \alpha, \beta, t}$. Observation~\ref{obs:rb-transitive} implies that $G^{\cF, \alpha, \beta, t}$ is $q$-red-blue-transitive for $q = \left(\frac{2\alpha}{\beta^2} \cdot k\right)$. Furthermore, the bound from previous paragraph is equivalent to $|E^{\cF, \alpha, \beta, t}_b| \geq \left(\frac{\delta}{4t^2}\right)k^2$. Applying Lemma~\ref{lem:rb-subgraph} to $G^{\cF, \alpha, \beta, t}$ with $d = \left(\frac{\delta}{8t^2}\right)k$ implies that there exists $\cS' \subseteq \cS$ of size at least $d$ such that
\begin{align} \label{eq:agr1}
\frac{|\{(S_1, S_2) \in \cS' \times \cS' \mid \{S_1, S_2\} \notin E^{\cF, \alpha, \beta, t}_r|}{|\cS'|^2} \geq 1 - \frac{q k}{d^2} = 1 - \frac{2048 t^8 \alpha}{\delta^4}.
\end{align}
From Observation~\ref{obs:non-red-consistent}, the left hand side of~\eqref{eq:agr1} is a lower bound for $\Pr_{S_1, S_2 \in \cS'}[\disagr(f_{S_1}, f_{S_2}) \leq (\rho\eta) \cdot n]$. Hence, we may invoke Lemma~\ref{lem:majority} on $\cF'$, which gives a function $g: [n] \to \{0, 1\}$ such that
\begin{align*}
\EX_{S \in \cS'}\left[\disa(g, f_S)\right] \leqs n \sqrt{\rho \left(\frac{2048 t^8\alpha}{\delta^4}\right) + \rho\eta} = n \cdot \sqrt{\rho} \cdot \sqrt{\frac{2048 t^8\alpha}{\delta^4} + \eta}
\end{align*}
This concludes our proof.
\end{proof}

\section{Soundness Analysis I: Generic Guarantee}
\label{sec:soundness-i}

In this section, we will use the main agreement testing theorem from the previous section to analyze the soundness of our reduction (Definition~\ref{def:main-red}). The analysis of this section will be of ``generic form'', in the sense that we will not plug in the parameters yet. These parameters for random subsets of clauses will be calculated and incorporated in to the soundness analysis in the next two sections.

For a 3-CNF formula $\Phi$ and a collection $\cT$ of subsets of its clauses, we write $\cS_{\Phi, \cT}$ to denote the collection $\{\var(T)\}_{T \in \cT}$. The generic soundness of our reduction can be stated as follows:

\begin{theorem} \label{thm:soundness-gen}
For any $0 < \eta, \rho, \alpha, \gamma, \mu, \delta < 1$ and $t, k, n, d \in \N$ such that $t \geq 2$ and $k \geq 10t/\alpha$, let $\cT = \{T_1, \dots, T_k\}$ be any collection of clauses of $\Phi$ satisfying the following properties:
\begin{itemize}
\item For every $i \ne j \in [k]$, $|\var(T_i) \cap \var(T_j)| \leq \rho n$.
\item For every $i \ne j \in [k]$, $(\cS_{\Phi, \cT} \setminus \{\var(T_i), \var(T_j)\})|_{\var(T_i) \cap \var(T_j)}$ is an $(\frac{\alpha}{(10t)^{2t}} \cdot k^{2t - 3}, 2t - 3, \eta)$-strong intersection disperser.
\item Any subcollection $\tcT \subseteq \cT$ of size at least $\frac{\delta k}{8 t^2}$ is $(\gamma, \mu)$-uniform. 
\end{itemize}

If $\val(\Phi) < 1 - \mu - (3\Delta/\gamma)\cdot \sqrt{\rho} \cdot \sqrt{\frac{2048 t^8\alpha}{\delta^4} + \eta}$, then $\wval(\cL_{\Phi, \cT, t}) < \delta$.
\end{theorem}

The above theorem follows almost immediately from the agreement testing theorem, as outlined in Section~\ref{sec:soundness-v-agreement}. The only additional ingredient needed here is the following lemma from~\cite{DinurM18}, which states that a global assignment that mostly agree with many local assignments must violate few clauses, provided that the collection of clauses considered are sufficiently uniform:

\begin{lemma}[{\cite[Lemma 33]{DinurM18}}] \label{lem:decode-good-assignment}
Let $\cT^*$ be any $(\gamma, \mu)$-uniform collection of subsets of clauses and $\sigma$ be any labeling of $\cT^*$. If there exists $\psi: \cX \to \{0, 1\}$ such that $\E_{T \in \cT^*}[\disagr(\phi, \sigma_T)] \leq \nu n$, then $\val(\psi) \geq 1 - \mu - 3\nu\Delta/\gamma$.
\end{lemma}

\begin{proof}[Proof of Theorem~\ref{thm:soundness-gen}]
Suppose contrapositively that $\wval(\cL_{\Phi, \cT, t}) \geq \delta$. In other words, there exists a left labeling $\sigma = (\sigma_T)_{T \in \cT}$ with $\wval(\sigma) \geq \delta$. By viewing $\sigma$ as collection $\cF = \{f_S\}_{S \in \cS_{\Phi, \cT}}$ where $f_{\var(T)} = \sigma_T$, we have $t$-$\wagr(\cF) \geq \delta$.

From Theorem~\ref{thm:agr-main}, we can find a subcollection $\tcT \subseteq \cT$ of size at least $\frac{\delta k}{8t^2}$ and a function $\psi: \cX \to \{0, 1\}$ such that 
\begin{align*}
\E_{S \in \cS_{\Phi, \tcT}}[\disagr(\psi, f_S)] \leq n \cdot \sqrt{\rho} \cdot \sqrt{\frac{2048 t^8\alpha}{\delta^4} + \eta}.
\end{align*}
From the third assumption, $\tcT$ is $(\gamma, \mu)$-uniform. Hence, we may apply Lemma~\ref{lem:decode-good-assignment}, which implies that $\val(\psi) \geq 1 - \mu - (3\Delta/\gamma) \cdot \sqrt{\rho} \cdot \sqrt{\frac{2048 t^8\alpha}{\delta^4} + \eta}$. This concludes our proof.
\end{proof}

\section{Parameters of Random Subsets of Clauses}
\label{sec:random}

In this section, we will calculate the parameters for the bounds needed in Theorem~\ref{thm:soundness-gen} (and more) if we pick each subset in $\cT$ by independently including each clause in it with probability $p$.

Throughout this section and the next section, we let $\Phi$ be any 3-CNF formula with $n$ variables and $m$ clauses such that each variable appears in at most $\Delta$ clauses and at least one clause. We use $\cC$ and $\cX$ to denote the set of clauses and the set of variables of $\Phi$ respectively. Moreover, we let $\cT = \{T_1, \dots, T_k\}$ be a collection of subsets of clauses selected by including each clause to each subset independently with probability $p$. 

We say that an event occurs with high probability (w.h.p.) if the probability that it occurs approaches one as $n \to \infty$ (or equivalently $m \to \infty$).

For a subset of variables $\cX_0 \subseteq \cX$, we use $\cC_{\cX_0}$ to denote the collection of clauses with at least one variable from $\cX_0$, i.e., $\cC_{\cX_0} := \{C \in \cC \mid \var(C) \cap \cX_0 \ne \emptyset\}$.

\subsection{Subset Size Bound}

We start with the easiest property to prove: the bound on the size of each $T \in \cT$. While this is not needed in Theorem~\ref{thm:soundness-gen}, it will be needed to bound the size of our resulting \textsc{Label Cover} instance.

\begin{proposition} \label{prop:subset-size}
W.h.p., $|T| \leq 2pm$ for all $T \in \cT$. Consequently, $|S| \leq 6p\Delta n$ for all $S \in \cS_{\Phi, \cT}$.
\end{proposition}

\begin{proof}[Proof of Proposition~\ref{prop:subset-size}]
Let us fix a subset $T_i \in \cT$. For every clause $C$, let $X_C$ denote the indicator variable whether $C$ belongs to $T_i$. $X_C$'s are i.i.d. Bernoulli random variable with mean $p$. Hence, by Chernoff bound, we have $$\Pr[|T_i| > 2pm] = \Pr[\sum_{C \in \cC} X_C > 2pm] = \exp(-pm).$$ By taking union bound over all $T_i \in \cT$, we get the desired result.
\end{proof}

\subsection{Pairwise Intersection}

Another property that is also simple to prove using Chernoff bound is that the intersection size $|\var(T_i) \cap \var(T_j)|$ is $O_{\Delta}(p^2 n)$, as stated below.

\begin{proposition} \label{prop:pairwise-int}
With high probability, for every $i \ne j \in [k]$, $|\var(T_i) \cap \var(T_j)| \leq 18p^2\Delta^2n$.
\end{proposition}

\begin{proof}[Proof of Proposition~\ref{prop:pairwise-int}]
We will prove the statement for a fixed pair $i, j$; union bound over $O(k^2)$ such pairs yields the desired result.

Let $\cX_i = \var(T_i)$. From Proposition~\ref{prop:subset-size}, we have $|\cX_i| \leq 6p\Delta n$. This implies that $|\cC_{\cX_i}| \leq 6p\Delta^2 n$. We will now bound $T_i \cap \cC_{\cX_i}$. For every $C \in \cC_{\cX_i}$, let $X_C$ denote the indicator variable whether $C$ is included in $T_j$. $X_C$'s are simply i.i.d. Bernoulli random variables with mean $p$. As a result, by Chernoff bound, we have $$\Pr[|T_j \cap \cC_{\cX_i}| > 6p^2\Delta^2n] = \Pr[\sum_{C \in \cC_{\cX_i}} X_C > 6p^2\Delta^2 n] \leq \exp(-p^2 n).$$
In other words, w.h.p., we have $|T_j \cap \cC_{\cX_i}| \leq 6p^2\Delta^2n$. Finally, observe that $|\var(T_i) \cap \var(T_j)| \leq 3|T_j \cap \cC_{\cX_i}|$. This yields the desired bound.
\end{proof}

\subsection{Uniformity}

Next, we state the parameters for uniformity property. Since this exact same lemma was proved before (as Lemma 36) in~\cite{DinurM18-arxiv}, we will not repeat the proof here; note that the proof is once again a simple Chernoff bound argument.

\begin{proposition}[{\cite[Lemma 36]{DinurM18-arxiv}}] \label{prop:random-uniform}
For any $\mu > 0$, any subcollection $\tcT \subseteq \cT$ of size at least $\lceil 8\ln(2/\mu) / p \rceil$ is w.h.p. $(p/2, \mu)$-uniform (with respect to the universe $\cC$).
\end{proposition}

\subsection{Strong Intersection Disperser}
\label{sec:strong-int-disp}

The last, and most challenging to prove, property is the parameters for strong intersection dispersers. A formal and quantitative version of the statement can be found below. 

\begin{lemma} \label{lem:strong-int-disp-clauses}
For any constant $\ell \in \N \cup \{0\}$ and $\kappa > 0$, for every $i \ne j \in [k]$, we have $(\cS_{\Phi, \cT} \setminus \{\var(T_i), \var(T_j)\})|_{\var(T_i) \cap \var(T_j)}$ is a $(\kappa (k - 2)^\ell, \ell, 6 \Delta \ell \cdot e^{-p \kappa (k - 2) / \ell})$-strong intersection disperser w.h.p.
\end{lemma}

The proof is divided into two parts. First, in Section~\ref{sec:random-int-disperser}, we argue that random subsets of a universe form strong intersection dispersers with good parameters; note that this section does \emph{not} deal with 3-CNF formulae at all. Then, in Section~\ref{sec:strong-int-3cnf}, we show how to adapt this bound to give the desired result on the collection $\cS_{\Phi, \cT}$.

\subsubsection{Random Subsets are Strong Intersection Dispersers}
\label{sec:random-int-disperser}

We will show that, a collection of random subsets are w.h.p. strong intersection dispersers with good parameters. The qualitative and formal statement is presented below.

\begin{lemma} \label{lem:strong-int-disp-random}
Let $\cU$ be any universe of size $m_0$, and let $0 < p, \kappa < 1$ and $\ell \in \N$ be any constants. Let $\cH = \{H_1, \dots, H_{k_0}\}$ be a set system generated by including each element $u \in \cU$ in each set $H_j$ with probability $p$. Then, $\cH$ is a $(\kappa k_0^\ell, \ell, 2 \ell \cdot e^{-p \kappa k_0/\ell})$-strong intersection disperser with high probability (as $m_0 \to \infty$).
\end{lemma}

It turns out that proving that random subsets form a strong intersection disperser is almost equivalent to the following natural question about monotone DNF formulae: for any monotone DNF formula with width (at most) $\ell$ and size $s$ on $k$ variables, what is the maximum probability that it evaluates to false under the $p$-biased distribution?

Recall that a \emph{monotone DNF formula} of \emph{maximum width} $w$ and \emph{size} $s$ on $k$ \emph{variables} is an AND of $s$ distinct terms, each of which is an OR of at most $w$ variables. We use $\DNF_{s, w, k}$ to denote the set of all monotone DNF formula of size at least $s$ and maximum width $w$ on $k$ variables. 

Moreover, recall that, the $p$-biased distribution on $\{0, 1\}^k$, denoted by $\pi_p^{\otimes k}$, is the distribution where each $x_i$ is i.i.d. True (= 1) w.p. $p$ and False (= 0) w.p. $1 - p$.

We can prove the following lemma, which partially answers the question of above form (and suffices for proving Lemma~\ref{lem:strong-int-disp-random}).

\begin{lemma} \label{lem:monotone-dnf}
For any $\varepsilon, p \in (0, 1), k, \ell \in \mathbb{N}$ and $f \in \DNF_{\varepsilon k^\ell, \ell, k}$, we have
\begin{align*}
\Pr_{\bx \sim \pi_{p}^{\otimes k}}[f(\bx) = 0] \leq \ell \cdot (1 - p)^{\varepsilon k/\ell}.
\end{align*}
\end{lemma}

Note here that, up to the dependency on $\ell$, the above lemma is essentially optimal in the regime where $k \gg \ell, 1/\varepsilon$. Namely, we can select $C \varepsilon k$ variables and consider a DNF formula that consists of every possible OR (of width $\ell$) terms that contain at least one such variables. If we pick a sufficiently large constant $C = C(\ell)$, then the size of the formula will be at least $\varepsilon k^\ell$. Moreover, the formula always evaluate to false if all these variables are set to false, which happens with probability $(1 - p)^{C_\ell \varepsilon k}$. This matches the bound we give in Lemma~\ref{lem:monotone-dnf} up to the dependency on $\ell$.

Before we prove Lemma~\ref{lem:monotone-dnf}, let us (briefly) argue why it immediately implies Lemma~\ref{lem:strong-int-disp-random}.

\begin{proof}[Proof of Lemma~\ref{lem:strong-int-disp-random}]
Let $h = \kappa k_0^\ell$. We will bound the probability that a particular set of subcollections $\cH_1, \dots, \cH_h$ violates~\eqref{eq:int-disperser} and use union bound over all such tuples at the end.

Let us fix $h$ distinct subcollections $\cH_1, \cdots, \cH_h \subseteq \cH$ each of size at most $\ell$. To calculate the probability that it violates~\eqref{eq:int-disperser}, let us construct a monotone DNF $f$ on variables $X_1, \dots, X_k$ by
\begin{align*}
f(X_1, \dots, X_k) = \bigvee_{j \in [h]} \left(\bigwedge_{T_i \in \cT_j} X_i\right).
\end{align*}
Clearly, $f$ has width (at most) $\ell$. Moreover, the size of $f$ is at least $h$, since each term is distinct.

Now, consider each element $u \in U$. For each $j \in [h]$, let $X_j(u)$ denotes the indicator variable whether $u$ belongs to $T_j$. Observe that $u$ belongs to $\bigcup_{j \in [h]}\left(\bigcap_{H \in \cH_i} H\right)$ iff $f(X_1(u), \cdots, X_{k_0}(u)) = 1$. Since $u$ is included in each $H_j$ independently with probability $p$, $(X_1(u), \cdots, X_{k_0}(u))$ is distributed as $\pi_p^{\otimes k_0}$. As a result, the probability that $u$ is not included in $\bigcup_{j \in [h]}\left(\bigcap_{H \in \cH_i} H\right)$ is exactly equal to $\Pr_{\bx \sim \pi_p^{\otimes k_0}}[f(\bx) = 0]$, which from Lemma~\ref{lem:monotone-dnf} is at most $\ell \cdot (1 - p)^{\kappa k_0/\ell} \leq \ell \cdot e^{- p \kappa k_0/\ell}$.

Let $Y_u$ denote the event that $u$ does not belong to $\bigcup_{j \in [h]}\left(\bigcap_{H \in \cH_i} H\right)$. Since $Y_u$'s are independent boolean random variables with mean at most $\ell \cdot e^{- p \kappa k_0/\ell}$, we may apply Chernoff bound, which implies that
\begin{align*}
\Pr\left[\left|U \setminus \bigcup_{j \in [h]}\left(\bigcap_{H \in \cH_i} H\right)\right| > 2 \ell \cdot e^{-p \kappa k_0/\ell} \cdot m_0\right] 
&= \Pr\left[\sum_{u \in U} Y_1 + \cdots + Y_u > (2\ell \cdot e^{- p \kappa k_0/\ell}) \cdot m_0\right] \\ 
&< \exp(-(2\ell \cdot e^{- p \kappa k_0/\ell}) \cdot m_0) \\
&= o(1). 
\end{align*}
In other words, a single $(\cH_1, \cdots, \cH_h)$ violates~\eqref{eq:int-disperser} with probability $o(1)$. There are at most $2^{kh}$ such $(\cH_1, \cdots, \cH_h)$'s. As a result, by union bound, we have that~\eqref{eq:int-disperser} holds for all $(\cH_1, \dots, \cH_h)$ with probability $1 - o(1)$.
\end{proof}

We now proceed to prove Lemma~\ref{lem:monotone-dnf} by induction on $\ell$.

\begin{proof}[Proof of Lemma~\ref{lem:monotone-dnf}]
\textbf{Base Case.} Let us consider the case $\ell = 1$. If $f$ contains an empty term, then obviously $\Pr_{\bx \sim \pi_p^{\otimes k}}[f(\bx) = 0] = 0$. On the other hand, if $f$ is an OR of at least $\varepsilon k$ variables. Then, the probability that $f$ is false (w.r.t distribution $\pi_p^{\otimes k}$) is at most $(1 - p)^{\varepsilon k}$ as desired.

\textbf{Inductive Step.} Suppose that the statement is true for $\ell = \ell_0$. We will show that this is also true for $\ell = \ell_0 + 1$. Consider any $f \in \DNF_{\varepsilon k^{\ell_0 + 1}, \ell_0 + 1, k}$. Once again, if $f$ contains an empty term, then the statement is trivially true.

Otherwise, let us consider the following process, for $m = \lceil\varepsilon k/(\ell_0 + 1)\rceil$.
\begin{itemize}
\item Let $S$ denote the set of all terms in $f$.
\item For $j = 1, \dots, m$:
\begin{itemize}
\item Let $x_{i_j}$ denote a variable that appears in the maximum number of terms in $S$. (Ties can be broken arbitrarily.) 
\item Let $S_j$ denote the set of terms that contain $x_{i_j}$.
\item Update $S$ to be $S \setminus S_j$.
\end{itemize}
\end{itemize}
For notational convenience, let us rearrange the coordinates so that $i_j = j$ for all $j = 1, \dots, m$.

Notice that, in each step, we remove at most $k^{\ell_0}$ terms from $S$. Hence, at least $\left(\frac{\ell_0}{\ell_0 + 1}\right) \varepsilon k^{\ell_0 + 1}$ terms remain at the beginning of each iteration. This means that, for all $j \in [m]$, $S_j$ is of size at least $\left(\frac{\ell_0}{\ell_0 + 1}\right) \varepsilon k^{\ell_0}$.

For every $j \in [m]$, we create a set $T_j$ of OR terms by removing $x_j$ from each term in $S_j$. We can now lower bound $\Pr_{\bx \sim \pi_p^{\otimes k}}[f(\bx) = 1]$ as follows:
\begin{align}
\Pr_{\bx \sim \pi_p^{\otimes k}}[f(\bx) = 1] &\geq \Pr_{\bx \sim \pi_p^{\otimes k}}\left[\left(\bigvee_{j \in [m]} \left(\bigwedge_{C \in S_j} C\right)\right) = 1\right] \nonumber \\
&\geq \sum_{i \in [m]} \Pr_{\bx \sim \pi_p^{\otimes k}}\left[\left(\bigvee_{j \in [m]} \left(\bigwedge_{C \in S_j} C\right)\right) = 1 \wedge x_i = 1 \wedge x_{i - 1} = \cdots = x_1 = -1\right] \nonumber \\
&\geq \sum_{i \in [m]} \Pr_{\bx \sim \pi_p^{\otimes k}}\left[\left(\bigwedge_{C \in S_i} C\right) = 1 \wedge x_i = 1 \wedge x_{i - 1} = \cdots = x_1 = -1\right] \nonumber \\
&= \sum_{i \in [m]} \Pr_{\bx \sim \pi_p^{\otimes k}}\left[\left(\bigwedge_{C \in T_i} C\right) = 1 \wedge x_i = 1 \wedge x_{i - 1} = \cdots = x_1 = -1\right] \nonumber \\
&= \sum_{i \in [m]} \Pr_{\bx \sim \pi_p^{\otimes k}}\left[\left(\bigwedge_{C \in T_i} C\right) = 1\right] \cdot p(1 - p)^{i - 1}, \label{eq:tmp1}
\end{align}
where the last inequality uses the fact that each coordinate of $x$ is independent under $\pi_p^{\otimes k}$, and that $x_1, \dots, x_i$ do not appear in any term in $T_i$.

Now, recall that each $\left(\bigwedge_{C \in T_i} C\right)$ is simply a monotone DNF formula with maximum width $\ell_0 - 1$ and of size at least $\left(\left(\frac{\ell_0}{\ell_0 + 1}\right) \varepsilon\right) k^{\ell_0}$. Hence, by our inductive hypothesis, we have 
\begin{align*}
 \Pr_{\bx \sim \pi_p^{\otimes k}}\left[\left(\bigwedge_{C \in T_i} C\right) = 0\right] \leq \ell_0 \cdot (1 - p)^{\frac{\varepsilon k}{\ell_0 + 1}}.
\end{align*}
Plugging this back into~\eqref{eq:tmp1}, we get
\begin{align*}
\Pr_{\bx \sim \pi_p^{\otimes k}}[f(\bx) = 1] &\geq  \sum_{i \in [m]} \left(1 - \ell_0 \cdot (1 - p)^{\frac{\varepsilon k}{\ell_0 + 1}}\right) \cdot p (1 - p)^{i - 1} \\
&= \left(1 - \ell_0 \cdot (1 - p)^{\frac{\varepsilon k}{\ell_0 + 1}}\right) \cdot \left(\sum_{i \in [m]} p(1 - p)^{i - 1}\right) \\
&= \left(1 - \ell_0 \cdot (1 - p)^{\frac{\varepsilon k}{\ell_0 + 1}}\right) \cdot \left(1 - (1 - p)^m\right) \\
&\geq 1 - \ell_0 \cdot (1 - p)^{\frac{\varepsilon k}{\ell_0 + 1}} - (1 - p)^m \\
&\geq 1 - (\ell_0 + 1) \cdot (1 - p)^{\frac{\varepsilon k}{\ell_0 + 1}},
\end{align*}
where the last inequality uses the fact that $m \geq \varepsilon k / (\ell_0 + 1)$.

Hence, the statement also holds for $\ell = \ell_0 + 1$, which concludes our proof.
\end{proof}

\subsubsection{From Random Subsets to $\cS_{\Phi, \cT}$}
\label{sec:strong-int-3cnf}

We will now use the bound obtained for random subsets in the previous subsection to argue about the collection $\cS_{\Phi, \cT}$. To do so, we also need the following lemma which allows us to translate strong intersection dispersers from the collection $\cT$ of subsets of clauses to the collection $\cS_{\Phi, \cT}$ of subsets of variables; its proof is exactly the same as that of Lemma 32 in~\cite{DinurM18-arxiv} except we replace ``intersection disperser'' with ``strong intersection disperser'', and is hence omitted.

\begin{obs}[\cite{DinurM18-arxiv}] \label{obs:clauses-v-variables}
If $\cT$ is $(r, \ell, \eta)$-strong intersection disperser, $\cS_{\Phi, \cT}$ is $(r, \ell, 3\Delta\eta)$-strong intersection disperser.
\end{obs}

We are now ready to prove Lemma~\ref{lem:strong-int-disp-clauses}.

\begin{proof}[Proof of Lemma~\ref{lem:strong-int-disp-clauses}]
Let us fix $i, j \in [k]$. We will show that $(\cS_{\Phi, \cT} \setminus \{\var(T_1), \var(T_2)\})|_{\var(T_i) \cap \var(T_j)}$ is a $(\kappa (k - 2)^\ell, \ell, 6 \Delta \ell \cdot e^{-p \kappa (k - 2) / \ell})$-strong intersection disperser with high probability. By taking union bound over all $k^2$ choices of $i, j$, we get the desired statement. For notational convenience, we will only consider $i = k - 1$ and $j = k$; this is without loss of generality due to symmetry.

Let $\cX^*$ denote $\var(T_{k - 1}) \cap \var(T_k)$. Furthermore, for $i \in [k - 2]$, let $T^*_i = T_i \cap \cC_{\cX^*}$. Notice here that each $T^*_i$ is distributed as follows: each clause $C \in \cC_{\cX^*}$ is included in $T^*_i$ independently with probability $p$. Hence, by Lemma~\ref{lem:strong-int-disp-random}, we have $\cT^* := \{T^*_1, \dots, T^*_{k - 2}\}$ is w.h.p. a $(\kappa (k - 2)^\ell, \ell, 2\ell \cdot e^{-p \kappa (k - 2) / \ell})$-strong intersection disperser (with respect to the universe $\cC_{\cX^*}$). 

We may now translate from the universe $\cC_{\cX^*}$ to $\cX^*$ by applying Observation~\ref{obs:clauses-v-variables}. This implies that $\{\var(T^*_1) \cap \cX^*, \dots, \var(T^*_k) \cap \cX^*\}$ is a $(\kappa (k - 2)^\ell, \ell, 6 \Delta \ell \cdot e^{-p \kappa (k - 2) / \ell})$-strong intersection disperser with respect to the universe $\cX^*$. We conclude the proof by observing that $\{\var(T^*_1) \cap \cX^*, \dots, \var(T^*_{k - 2}) \cap \cX^*\} = (\cS_{\Phi, \cT} \setminus \{\var(T_{k - 1}), \var(T_k)\})|_{\cX^*}$.
\end{proof}

\section{Soundness Analysis II: Guarantee for Random Subsets}
\label{sec:soundness-ii}

We will now plug in the parameters from the previous section into Theorem~\ref{thm:soundness-gen} to give a soundness guarantee when $\cT$ consists of random subsets of clauses. The ready-to-use version of the theorem is stated and proved below.

\begin{theorem} \label{thm:soundness-random}
For any constants $\varepsilon, \delta > 0$ and $\Delta, t \in \N$ such that $t \geq 2$, there exists a constant $C = C(\varepsilon, \Delta, \delta, t)$ such that, if let $\cT$ be a collection of $k$ random subsets of clauses where each clause is included in each subset with probability $p = C/k$, then, for any sufficiently large $k$, the following holds with high probability: if $\val(\Phi) < 1 - \varepsilon$, then $\wval(\Phi_{\Phi, \cT, t}) < \delta$.
\end{theorem}

\begin{proof}
We pick $C$ to be $\left(\frac{100\Delta t}{\varepsilon \delta}\right)^{100t} \ln(\frac{\Delta t}{\varepsilon \delta})$. We will prove the statement for $k \geq C^{100}$.
Let us define the parameters as follows:
\begin{itemize}
\item $\mu = \varepsilon/2$.
\item $\gamma = p/2$.
\item $\kappa = \left(\frac{100\Delta t}{\varepsilon \delta}\right)^{-50t}$.
\item $\alpha = (10t)^{2t} \cdot \frac{\kappa \cdot (k - 2)^{2t - 3}}{k^{2t - 3}} < \frac{\varepsilon^2 \delta^4}{10^{10} \cdot \Delta^5 \cdot t^8}$.
\item $\rho = 18p^2 \Delta^2$.
\end{itemize}
For the selected parameters, we have the following:
\begin{itemize}
\item From Lemma~\ref{lem:strong-int-disp-clauses} with $\ell = 2t - 3$, we have: for every $i \ne j \in [k]$, w.h.p. $(\cS_{\Phi, \cT} \setminus \{T_i, T_j\})|_{\var(T_i) \cap \var(T_j)}$ is a $(\frac{\alpha}{(10t)^{2t}} \cdot k^{2t - 3}, 2t - 3, \eta)$ where
\begin{align*}
\eta &= 6\Delta(2t - 3) \cdot e^{-\frac{p\kappa (k - 2)}{2t - 3}} \\
(\text{From our choice of } p, \kappa) &< 6\Delta \cdot e^{-100\ln(\frac{\Delta r}{\varepsilon \delta})} \\
&< \frac{\varepsilon^2}{10 \cdot \Delta^4}.
\end{align*} 
\item From Proposition~\ref{prop:pairwise-int}, we have $|\var(T_i) \cap \var(T_j)| < \rho n$ for all $i \ne j \in [k]$ w.h.p.
\item From Proposition~\ref{prop:random-uniform}, we have that w.h.p. any subcollection $\tcT \subseteq \cT$ of size at least $\lceil 8\ln(2/\mu)/p \rceil \leq \frac{\delta k}{8 t^2}$ is $(\mu, p/2)$-uniform.
\end{itemize}
When all the high probability events above hold and $\wval(\Gamma_{\Phi, \cT, t}) \geq \delta$, we may apply Theorem~\ref{thm:soundness-gen} which implies
\begin{align*}
\val(\Phi) &\geq 1 - \mu - \frac{3\Delta}{\gamma} \cdot \sqrt{\rho} \cdot \sqrt{\frac{2048 t^8\alpha}{\delta^4} + \eta} \\
&\geq 1 - \frac{\varepsilon}{2} - \frac{3\Delta}{p/2} \cdot \sqrt{18p^2\Delta^2} \cdot \sqrt{\frac{2048 t^8\alpha}{\delta^4} + \eta} \\
&\geq 1 - \frac{\varepsilon}{2} - \sqrt{648 \Delta^4\cdot\left(\frac{2048 t^8}{\delta^4} \cdot \frac{\varepsilon^2 \delta^4}{10^{10} \cdot \Delta^4 \cdot t^8} + \frac{\varepsilon^2}{10^{10} \cdot \Delta^4}\right)} \\
&\geq 1 - \varepsilon,
\end{align*} 
which concludes our proof.
\end{proof}

\section{Inapproximability for Label Cover}
\label{sec:lb}

Our main hardness of \textsc{Label Cover} (Theorem~\ref{thm:main-label-cover-weak-agr}) now follows almost trivially from Theorem~\ref{thm:soundness-random}.

\begin{proof}[Proof of Theorem~\ref{thm:main-label-cover-weak-agr}]
Assume that Gap-ETH holds, and let $\varepsilon$ be the gap in Gap-ETH. Moreover, let $\delta > 0$ and $t \in \N$ be any constants.

Let $\Phi$ be any 3-CNF formula such that each variable appears in at most $\Delta$ clauses. Consider $\cL_{\Phi, \cT, t}$ where $\cT$ denote the collection of $k$ subsets of clauses where each clause is included in each subset independently with probability $p = \frac{C}{k}$ where $C$ is the constant from Theorem~\ref{thm:soundness-random}. If $\Phi$ is satisfiable, then $\cL_{\Phi, \cT, t}$ is clearly also satisfiable. On the other hand, if $\val(\Phi) < 1 - \varepsilon$, then Theorem~\ref{thm:soundness-random} implies that $\wval(\cL_{\Phi, \cT, t}) < \delta$ w.h.p. Finally, Proposition~\ref{prop:subset-size} implies that the size of $\cL_{\Phi, \cT, t}$ is only $2^{O_{\Delta, r, \delta, \varepsilon}(m/k)}$ w.h.p.

Hence, if there is an $T(k) \cdot N^{o(k)}$-time algorithm for distinguishing $\val(\cL) = 1$ and $\wval(\cL) < \delta$, the we can run this algorithm on $\cL_{\Phi, \cT, t}$; this distinguishes $\val(\Phi) = 1$ and $\val(\Phi) < 1 - \varepsilon$ in time $2^{o(m)}$, which would violate Gap-ETH.
\end{proof}

\section{From Label Cover to Other Problems}
\label{sec:consequences}

In this section, we will describe how our hardness of approximation for \textsc{Label Cover} proved in the previous section can be used to prove inapproximability of other problems mentioned in the introduction. For brevity, we use the following convention throughout this section: for any instance $\cI$ of an optimization problem $P$, $\opt_P(\cI)$ denote the optimum of $\cI$ with respect to $P$.

\subsection{Right Alphabet Reduction for Label Cover}

Our hardness of \textsc{Label Cover} is not yet in a ready-to-use form, because the reductions we will apply below (see Theorem~\ref{thm:max-cov-feige} and Theorem~\ref{thm:set-cov-feige}) blow the instance size up exponentially in terms of the right alphabet size $\max_{v \in V} |\Sigma_v|$. However, our result from the previous section does not have a bound on this value; in fact, our reduction produces a \textsc{Label Cover} instance with large right alphabet size ($\approx N^{O(1/k)}$). Nevertheless, this turns out not to be an issue, since the right alphabet size can be easily reduced while the left hand side (both the vertex set and alphabet sets) remains the same, as stated below.

\begin{lemma}
For any parameter $\delta > 0$, there is a polynomial time algorithm that, given a bi-regular label cover instance $\cL = (U, V, E, \{\Sigma_u\}_{u \in U}, \{\Sigma_v\}_{v \in V}, \{\pi_e\}_{e \in E})$ of size $N$ and right degree $t$, produces another bi-regular label cover instance $\cL' = (U, V', E', \{\Sigma_u\}_{u \in U}, \{\Sigma_{v'}\}_{v' \in V'}, \{\pi'_e\}_{e \in E'})$ with the same left vertices and alphabets such that 
\begin{itemize}
\item (Completeness) If $\cL$ is satisfiable, then $\cL'$ is also satisfiable.
\item (Soundness) If $\wval(\cL') \leq \wval(\cL) + \delta$.
\item (Right Alphabet Size) For all $v' \in V'$, $\Sigma_{v'} = O(t^2/\delta)$.
\item (Right Degree) The right degree of $(U, V', E')$ remains $t$. 
\end{itemize}
\end{lemma}

The proof proceeds by replacing each right alphabet with an error correcting code with distance $1 - \delta/t^2$; the point here is that, if a left labeling $\sigma = (\sigma_u)_{u \in U}$ does not weakly agree on $v \in V$, then every pair $u_1, u_2$ of $v$'s neighbors ``disagree'', i.e., $\pi_{(u_1, v)}(\sigma_{u_1}) \ne \pi_{(u_2, v)}(\sigma_{u_2})$. Since we are replacing $\Sigma_v$ with an error correcting code with distance $1 - \delta/t^2$, they will still disagree on all but $\delta/t^2$ fraction of the coordinates. This indeed ensures the soundness of the reduction. 
Below we use the Hadamard codes only because the relationship between their alphabet sizes and distances are simple. In general, one could use any code such that the relative distance is $1 - \Omega(1/q)$ where $q$ is the alphabet size.

\begin{proof}
We may assume w.l.o.g. that $\delta \geq 1/n$, as otherwise the alphabet size already satisfies $|\Sigma_v| = O(t^2/\delta)$ for all $v \in V$ and there is no need to modify the instance $\cL$ at all.

Let $R$ be $\max_{v \in V} |\Sigma_v|$, $q$ be the smallest prime such that $q \geq t^2 / \delta$ and $\ell = \lceil \log_q R \rceil$. Consider the Hadamard code $C: \mathbb{F}_q^\ell \to \mathbb{F}_q^{q^\ell}$ with alphabet size $q$, message length $\ell$, block length $q^\ell$ and relative distance $1 - 1/q$. We may associate each label in $\Sigma_v$ for each $v \in V$ with an element of $\mathbb{F}_q^{q^\ell}$. With this in mind, we can define our new label cover instance $\cL'$ as follows:
\begin{itemize}
\item Let $V' = V \times [q^t]$ and $\Sigma_{V'} = \mathbb{F}_q$.
\item We add edges in $E'$ between each $(v, j) \in V \times [q^t]$ to all neighbors $u \in U$ of $v$.
\item We define the constraint $\pi_{(u, (v, j))}$ by $$\pi_{(u, (v, j))}(\beta) := C(\pi_{(u, v)}(\beta))_j.$$
(In other words, we take the $j$ coordinate of the codeword for $\pi_{(u, v)}(\beta)$.)
\end{itemize}

It is obvious that the completeness and alphabet size properties are satisfied. We now argue the soundness property. Let us consider any left labeling $\sigma = (\sigma_u)_{u \in U}$. From definition of $\wval(\cL)$, at most $\wval(\cL)$ fraction of right vertices are weakly agreed on by $\sigma$. Let us now consider any vertex $v$ not agreed on by $\sigma_V$ (with respect to $\cL$); this implies that any two distinct neighbors $u_1, u_2$ of $v$ have $\pi_{(u, v)}(\sigma_{u_1}) \ne \pi_{(u, v)}(\sigma_{u_2})$. For a fixed $v$ and a pair of such $(u_1, u_2)$, since the error correcting code $C$ has distance $1 - 1/q$, $C(\pi_{(u_1, v)}(\sigma_{u_1}))_j = C(\pi_{(u_2, v)}(\sigma_{u_2}))_j$ only for at most $1/q$ fraction of $j \in [q^t]$. Since there are only $t^2$ such pairs $(u_1, u_2)$ for such fixed $v$, $(v, j)$ will be weakly agree on by $\sigma$ with respect to $\cL'$ on at most $\frac{t^2}{q} \leq \delta$ fraction of $j \in [q^t]$. 

As a result, the fraction of vertices in $V'$ that can be weakly agreed on by $\sigma$ in $\cL'$ is at most $\wval(\cL) + \delta$ as desired. This indeed implies that $\wval(\cL') \leq \wval(\cL) + \delta$.
\end{proof}

By applying the above lemma to Theorem~\ref{thm:main-label-cover-weak-agr}, we can get a refinement of the hardness of approximation for \textsc{Label Cover}, where the right alphabet size is $O(t^2/\delta)$, as stated below.  

\begin{lemma} \label{lem:label-cover-bounded-alphabet}
Assuming Gap-ETH, the following holds. For every constant integer $t \geq 2$ and $\delta > 0$, and any function $T$, there is no $T(k) \cdot N^{o(k)}$-time algorithm that can, given a label cover instance $\cL = (U, V, E, \{\Sigma_u\}_{u \in U}, \{\Sigma_v\}_{v \in V}, \{\pi_e\}_{e \in E})$ of size $N$ with $k = |U|$ such that the constraint graph $(U, V, E)$ is bi-regular with right-degree $t$, distinguish between the following two cases:
\begin{itemize}
\item (Completeness) $\val(\cL) = 1$.
\item (Soundness) $\wval(\cL) < \delta$.
\item (Bounded Right Alphabet) For every $v \in V$, $|\Sigma_v| < O(t^2/\delta)$.
\end{itemize}
\end{lemma}

\subsection{Max $k$-Coverage}

We now move on to specify how our hardness of approximating \textsc{Label Cover} from the previous section can be used to improve running time lower bounds for inapproximability results of other problems; as stated earlier, these are merely consequences of known reductions, which are not a contribution of this work. 

We start with the \textsc{Max $k$-Coverage} problem; Feige's reduction that provide NP-hardness of approximating the problem can be reformulated as follows.

\begin{theorem}[\cite{Feige98}] \label{thm:max-cov-feige}
For any constant $\varepsilon > 0$, there exist constants $\delta > 0, t \in \N$ and a polynomial-time reduction that takes in a \textsc{Label Cover} bi-regular instance $\cL = (U, V, E, \{\Sigma_u\}_{u \in U}, \{\Sigma_v\}_{v \in V}, \{\pi_e\}_{e \in E})$ with right degree $t$, and output an instance $\cI' = (U', \cS, k)$ of \textsc{Max $k$-Coverage} such that
\begin{itemize}
\item (Completeness) If $\cL$ is satisfiable, then $\opt_{\textsc{Max }k\textsc{-Coverage}}(\cI') = |U'|$.
\item (Soundness) If $\wval(\cL) < \delta$, then $\opt_{\textsc{Max }k\textsc{-Coverage}}(\cI') < (1 - 1/e + \varepsilon) \cdot |U'|$.
\item (Size Bound) $|\cS| \leq N_{\cL}$ and $|U'| \leq \sum_{v \in V} t^{|\Sigma_v|}$.
\item (Parameter) The parameter $k$ is preserved in the reduction.
\end{itemize} 
\end{theorem}

We may apply the above reduction to our hardness of \textsc{Label Cover} in Lemma~\ref{lem:label-cover-bounded-alphabet}. Notice that, while the universe size $|U'|$ of the resulting instance $\cS$ can blow up exponentially in $|\Sigma_v|$, the bound on the right alphabet size in Lemma~\ref{lem:label-cover-bounded-alphabet} ensures that $|\Sigma_v| = O(t^2/\delta)$ and hence $|U'| \leq O_{t, \delta}(N_{\cL})$. Thus, we arrive at the desired hardness of \textsc{Max $k$-Coverage} (Corollary~\ref{cor:max-k-cov}).

\subsection{$k$-Median and $k$-Mean}

Guha and Khuller~\cite{GuhaK99} gives a reduction from \textsc{Max $k$-Coverage} to the \textsc{Uncapacitated Facility Location} problem, which is also applicable to the \textsc{$k$-Median} and \textsc{$k$-Mean} problems. In particular, their reduction lets the elements in the universe be the clients and the subsets be the the facility; the metric is then defined so that, for an element $u$ and a subset $S$, $d(u, S) = 1$ iff $u \in S$ and $d(u, S) = 3$ otherwise. Selecting $k$ facilities is exactly the same as selecting $k$ subsets in the \textsc{Max $k$-Coverage} problem. In the completeness case where there are $k$ subsets that cover the whole universe, these subsets give a solution to \textsc{$k$-Median}/\textsc{$k$-Mean} such that each client pays the cost of 1. On other hand, in the soundness case, each client not covered by the subsets will have to pay 3 for \textsc{$k$-Median} and 9 for \textsc{$k$-Mean}. To summarize, Guha and Khuller's reduction has the following properties:

\begin{lemma}[\cite{GuhaK99}]
There is a polynomial-time reduction that takes in an instance $(U, \cS, k)$ of \textsc{Max $k$-Coverage} and produces an instance $(V, F, k)$ of \textsc{$k$-Median} and \textsc{$k$-Mean} such that
\begin{itemize}
\item If $\frac{\opt_{\textsc{Max }k\textsc{-Coverage}}(U, \cS, k)}{|U|} = 1 - \tau$, then $\frac{\opt_{k\textsc{-Median}}(V, F, k)}{|V|} = 1 + 3\tau$ and $\frac{\opt_{k\textsc{-Mean}}(V, F, k)}{|V|} = 1 + 9\tau$.
\item The parameter $k$ is preserved in the reduction.
\end{itemize}
\end{lemma}

By applying the above reduction to Corollary~\ref{cor:max-k-cov}, we immediately arrive at the hardness of approximation for \textsc{$k$-Median} and \textsc{$k$-Mean} (Corollary~\ref{cor:mean-median}).

\subsection{$k$-Unique Set Cover}

For the $k$-\textsc{Unique Set Cover} problem, we once again resort to the reduction of Feige~\cite{Feige98}, which is the same as the one in Theorem~\ref{thm:max-cov-feige} above. In terms of \textsc{Set Cover}, the reduction yields the following\footnote{Note that, in Feige's original work~\cite{Feige98}, the reduction is not stated in this term, but one can find such a statement in, e.g.,~\cite{KLM18}.}:

\begin{theorem}[\cite{Feige98}] \label{thm:set-cov-feige}
For any constants $C > 1$ and any integer $t \geq 2$, there exist constants $\delta > 0$ and a polynomial-time reduction that takes in a \textsc{Label Cover} instance $\cL = (U, V, E, \{\Sigma_u\}_{u \in U}, \{\Sigma_v\}_{v \in V}, \{\pi_e\}_{e \in E})$ with right degree $t$, and output an instance $\cI' = (U', \cS, k)$ of \textsc{$k$-Unique Set Cover} such that
\begin{itemize}
\item (Completeness) If $\cL$ is satisfiable, then there exist $k$ subsets from $\cS$ that covers each element of $U'$ exactly once.
\item (Soundness) If $\wval(\cL) < \delta$, then $\opt_{\textsc{Set Cover}}(\cI) > C \cdot k$.
\item (Size Bound) $|\cS| \leq N_{\cL}$ and $|U'| \leq \sum_{v \in V} t^{|\Sigma_v|}$.
\item (Parameter) The parameter $k$ is preserved in the reduction.
\end{itemize} 
\end{theorem}

By applying the reduction to Lemma~\ref{lem:label-cover-bounded-alphabet} with some fixed constant $t$ (say $t = 2$), we arrive at the claimed hardness of approximation for $k$-\textsc{Unique Set Cover} (Corollary~\ref{cor:k-unique-cov}). We remark here that, for $k$-\textsc{Unique Set Cover} (and subsequent problems), we only need the hardness of \textsc{Label Cover} for a particular $t$ (e.g. $t = 2$) instead of \emph{for all} constant $t$, as required for \textsc{Max $k$-Coverage}. Indeed, this is the reason why, in the author's thesis~\cite{M-phd-thesis}, the result from~\cite{DinurM18} can still be used to prove parameterized inapproximability for $k$-\textsc{Unique Set Cover}, despite the fact that the hardness of approximation for \textsc{Label Cover} in~\cite{DinurM18} is achieved only for $t = 2$. However, as we mentioned earlier, the hardness in~\cite{DinurM18} does not yield an $T(k) \cdot N^{\Omega(k)}$ running time lower bound even for $t = 2$, and hence the additional proof modifications in this paper are still needed to achieve the tight running time lower bounds.

\subsection{$k$-Nearest Codeword Problem and $k$-Nearest Vector Problem}

Arora \etal~\cite{ABSS97} gives a reduction from \textsc{Unique Set Cover} to \textsc{Nearest Codeword} and \textsc{Nearest Vector}. The properties of the reduction are stated below.

\begin{theorem}[\cite{ABSS97}] \label{thm:abss}
For any constant $p \geq 1$, there is a polynomial-time reduction that takes in an instance $(U, \cS, k)$ of $k$-\textsc{Unique Set Cover} and produces instances $(\bA, \bx, k)$ of $k$-\textsc{NCP} and $(\bA', \bx', k)$ of $k$-\textsc{CVP}$_p$ such that
\begin{itemize}
\item (Completeness) If there exists $k$ subsets in $\cS$ that uniquely covers $U$, then $\opt_{\textsc{NCP}}(\bA, \bx) = \opt_{\textsc{CVP}_p}(\bA', \bx') = k$
\item (Soundness) If no $t$ subsets in $\cS$ that covers $U$, then $\opt_{\textsc{NCP}}(\bA, \bx) = \opt_{\textsc{NVP}_p}(\bA', \bx') > t$.
\item The parameter $k$ is preserved.
\end{itemize}
\end{theorem}

By applying the above reduction to our hardness of approximating $k$-\textsc{Unique Set Cover}, we immediately arrive at the desired inapproximability results for \textsc{$k$-Nearest Codeword} or \textsc{$k$-Nearest Vector} (Corollaries~\ref{cor:ncp} and~\ref{cor:nvp}).

\subsection{$k$-Minimum Distance Problem}

The reduction of Bhattacharyya \etal~\cite{BGKM18} from $k$-\textsc{NCP} to $k$-\textsc{MDP} can be summarized as follows. (Note that, in the conference version of the paper, $k'$ is not chosen to be $O_{\varepsilon}(k)$. However, it is easy to see that such a parameter setting is possible. For the formulation more similar to the one below, please see Lemma 10.6 in~\cite{M-phd-thesis}.)

\begin{theorem}[\cite{BGKM18}]
For every $\varepsilon > 0$, there exists a constant $C_{\varepsilon} > 1$ and a (randomized) polynomial-time reduction that takes in an instance $(\bA, \bx, k)$ of $k$-\textsc{Nearest Codeword} and produces an instance $(\bB, k')$ of $k'$-\textsc{Minimum Distance} such that
\begin{itemize}
\item (Completeness) If $\opt_{\textsc{NCP}}(\bA, \bx) \leq k$, then $\opt_{\textsc{MDP}}(\bB) \leq k'$ with probability at least $\frac{1}{k^{O(k)}}$.
\item (Soundness) If $\opt_{\textsc{NCP}}(\bA, \bx) > C_{\varepsilon} \cdot k$, then $\opt_{\textsc{MDP}}(\bB) > (2 - \varepsilon) \cdot k'$.
\item The new parameter $k'$ is $O_{\varepsilon}(k)$
\end{itemize}
\end{theorem}

By applying the above reduction to our parameterized inapproximability of $k$-\textsc{NCP} (Corollary~\ref{cor:ncp}), we immediately arrive at our hardness for $k$-\textsc{MDP} (Corollary~\ref{cor:mdp}).

\subsection{$k$-Shortest Vector Problem}

Although Khot~\cite{Khot05} originally gives a reduction from \textsc{CVP} to \textsc{SVP} in the non-parameterized setting, it was observed in~\cite{BGKM18} that the reduction also works in the parameterized setting as well, as stated below:

\begin{theorem}[\cite{Khot05}] \label{thm:khot-svp}
For every $p > 1$, there exist constants $C_p > 1, \delta_p > 0$ and a (randomized) polynomial-time reduction that takes in an instance $(\bA, \bx, k)$ of $k$-\textsc{CVP$_p$} and produces an instance $(\bB, k')$ of $k'$-\textsc{SVP$_p$} such that
\begin{itemize}
\item (Completeness) If $\opt_{\textsc{CVP}_p}(\bA, \bx) \leq k$, then $\opt_{\textsc{SVP}_p}(\bB) \leq k'$ with probability at least $0.9$.
\item (Soundness) If $\opt_{\textsc{CVP}_p}(\bA, \bx) > C_p \cdot k$, then $\opt_{\textsc{SVP}_p}(\bB) > (1 + \delta_p) \cdot k'$ with probability at least $0.9$.
\item The new parameter $k'$ is $O_{p}(k)$
\end{itemize}
\end{theorem}

Theorem~\ref{thm:khot-svp} and Corollary~\ref{cor:nvp} immediately imply Corollary~\ref{cor:svp}.

\section{Conclusion and Open Questions}
\label{sec:open}

In this work, we have shown a strong running time lower bound for inapproximability of the \textsc{Label Cover} problem, where we view the number of left vertices $|U|$ as the parameter. As immediate consequences, we have obtained strong running time lower bounds for inapproximability of \textsc{Max $k$-Coverage}, \textsc{$k$-Unique Set Cover} and other related problems. 

Despite the progress made in this paper, many aspects of parameterized inapproximability of \textsc{Label Cover} remains open. Firstly, we have only managed to show hardness of approximation of the problem under the Gap-ETH assumption. On the other hand, there has been many parameterized hardness of approximation results shown under more standard (non-gap) assumptions. Hence, a main question here is whether the inapproximability of \textsc{Label Cover} can be shown using a non-gap assumption such as W[1] $\ne$ FPT or ETH:

\begin{problem}
Is it W[1]-hard or ETH-hard to approximate \textsc{Label Cover} to within some constant factor?
\end{problem}

We remark here that this question is equivalent to the question of whether it is W[1]- or ETH-hard to approximate \textsc{2-CSP} (parameterized by the number of variables) to within some constant factor\footnote{Since \textsc{Label Cover} is a special case of \textsc{2-CSP}, this direction of the reduction is obvious. The other reduction is from the so-called clause-variable game (see, e.g.~\cite{AIM14}).}; Lokshtanov~\cite{LRSZ17} conjectured that the answer is positive and name the it the \emph{Parameterized Inapproximability Hypothesis (PIH)}. It is also worthwhile to point out that, if this conjecture is true, then it is also hard to approximate $k$-\textsc{Clique} to within any constant factor (under the same assumption); currently $k$-\textsc{Clique} is only known to be hard to approximate under Gap-ETH~\cite{CCKLM17}.

The second question is regarding the running time lower bound if we change the parameter. An interesting parameter here is to consider the total number of vertices $|U| + |V|$, instead of the number of left vertices. Notice that, in our reduction, $|V|$ is $\Theta(|U|^t)$; hence, the running time lower bound we can get is only of the form $T(k) \cdot N^{\Omega((|U| + |V|)^{1/t})}$. Hence, we may ask the following:

\begin{problem}
Is there an $T(|U| + |V|) \cdot N^{o(|U| + |V|)}$-time constant factor approximation algorithm for \textsc{Label Cover}?
\end{problem}

Interestingly, the lower bound of the form $T(|U| + |V|) \cdot N^{\Omega(|U| + |V|)}$ is not known even for the \emph{exact} version of the problem. This question is closely related to the following problem considered by Marx in~\cite{Marx10}: is there a $T(k) \cdot N^{o(k)}$-time algorithm for the \emph{exact} version of \textsc{2-CSP} where $k$ denote the number of constraints (i.e. edges)? Marx provided a nearly matching lower bound of $T(k) \cdot N^{\Omega(k/\log k)}$ for the problem; it remains open to close the gap of $\log k$ in the exponent. Note that, if the tight lower bound of $T(k) \cdot N^{\Omega(k)}$ can be achieved for the exact (resp. $(1 + \varepsilon)$-approximate) version of \textsc{2-CSP}, then $T(|U| + |V|) \cdot N^{\Omega(|U| + |V|)}$ running time lower bound also follows for the exact (resp. $(1 + \varepsilon')$-approximate) version of \textsc{Label Cover}. This follows from the well-known clause-variable reduction (see, e.g.~\cite{AIM14}). However, to the best of our knowledge, the other direction of the implication (from \textsc{Label Cover} to \textsc{2-CSP}) is unknown.

Regarding the individual problems we consider, there are many questions that are yet unanswered. Let us highlight a couple directions which we find interesting and touch upon more central issues that might be more useful in a wider context. First is the question of whether/how to prove a running time of the form $T(k) \cdot N^{o(k)}$ for \emph{any constant} factor inapproximability of $k$-\textsc{Minimum Distance} and $k$-\textsc{Shortest Vector}: 

\begin{problem}
Is there an $T(k) \cdot N^{o(k)}$-time constant factor approximation algorithm for $k$-\textsc{Minimum Distance} and $k$-\textsc{Shortest Vector}?
\end{problem}

All known inapproximability results (even NP-hardness) for $k$-\textsc{Minimum Distance} and $k$-\textsc{Shortest Vector} with large factors proceed by first proving a hardness for small factors and then amplify the gap via self-tensoring the instance repeatedly~\cite{DumerMS03,Khot05,HavivR12}. Similar to parallel repetition, self-tensoring blow up the parameter from $k$ to $k^2$; the running time lower bound will be only $T(k) \cdot N^{o(\sqrt{k})}$ even after one such operation. Hence, to prove lower bounds of the form $T(k) \cdot N^{o(k)}$, one likely has to come up with a different proof for a large factor inapproximability of both problems, such as a ``one-shot proof'' that gives large factors without tensoring.

We remark that such an issue is also present beyond the parameterized regime. Specifically, while it is possible to show (assuming Gap-ETH) that both problems are hard to approximate to within \emph{some} constant factor in $2^{o(N)}$-time~\cite{AggarwalS18,SV19}, it is unknown how to extend this to \emph{any} constant factor. The barrier here is exactly the same as in the parameterized regime described above, since self-tensoring also blows up the instance size from $N$ to $N^2$. Thus, it is possible that resolving the above (parameterized) question might help make progress on this problem as well.

Our final question is the question about FPT approximability of $k$-\textsc{Unique Set Cover}. Recall that a parameterized problem is said to be \emph{totally FPT inapproximable} if there is no FPT time algorithm that achieves $f(k)$-approximation for any function $f$. The question is whether this is the case for $k$-\textsc{Unique Set Cover}:

\begin{problem}
Is  $k$-\textsc{Unique Set Cover} totally FPT inapproximable?
\end{problem}

The problem is particularly interesting because the reduction from \textsc{Label Cover} \emph{cannot} give totally FPT inapproximability for the $k$-\textsc{Unique Set Cover} problem. A detailed explanation can be found in Section 9.8 of the author's PhD dissertation~\cite{M-phd-thesis}. A short version of this is that, we can always assume without loss of generality that the number of right vertices in a \textsc{Label Cover} instance is at most $2^k$, because we can merge two vertices together if they have the same set of neighbors. This implies that the soundness of $k$-\textsc{Unique Set Cover} in Feige's reduction is at most $2^k$. As a result, to prove total FPT inapproximability of $k$-\textsc{Unique Set Cover}, one has to start from a different problem than \textsc{Label Cover}; this might help us identify new primitives for parameterized hardness of approximation. On the other hand, if one hopes to answer the question negatively, one has to come up with a new approximation algorithm for \textsc{Unique Set Cover}, which could also be a potentially interesting direction for research.

Note that, due to the reduction of Arora et al.~\cite{ABSS97} (Theorem~\ref{thm:abss}), if the answer turns out to be ``yes'' for $k$-\textsc{Unique Set Cover}, then the same holds for $k$-\textsc{Nearest Codeword} and $k$-\textsc{Closest Vector}. However, it might be possible to prove that these two are totally FPT inapproximable without going through the \textsc{Unique Set Cover}. (Indeed, the inapproximability results shown in~\cite{BELM} for these two problems are \emph{not} via \textsc{Unique Set Cover}.)

\bibliographystyle{alpha}
\bibliography{citations}

\newcommand{\etalchar}[1]{$^{#1}$}
\begin{thebibliography}{KMN{\etalchar{+}}04}

\bibitem[ABSS97]{ABSS97}
Sanjeev Arora, L{\'{a}}szl{\'{o}} Babai, Jacques Stern, and Z.~Sweedyk.
\newblock The hardness of approximate optima in lattices, codes, and systems of
  linear equations.
\newblock {\em J. Comput. Syst. Sci.}, 54(2):317--331, 1997.

\bibitem[AGK{\etalchar{+}}04]{AryaGKMMP04}
Vijay Arya, Naveen Garg, Rohit Khandekar, Adam Meyerson, Kamesh Munagala, and
  Vinayaka Pandit.
\newblock Local search heuristics for $k$-median and facility location
  problems.
\newblock {\em {SIAM} J. Comput.}, 33(3):544--562, 2004.

\bibitem[AIM14]{AIM14}
Scott Aaronson, Russell Impagliazzo, and Dana Moshkovitz.
\newblock {AM} with multiple merlins.
\newblock In {\em {IEEE} 29th Conference on Computational Complexity, {CCC}
  2014, Vancouver, BC, Canada, June 11-13, 2014}, pages 44--55, 2014.

\bibitem[AK14]{AustrinK14}
Per Austrin and Subhash Khot.
\newblock A simple deterministic reduction for the gap minimum distance of code
  problem.
\newblock {\em {IEEE} Trans. Information Theory}, 60(10):6636--6645, 2014.

\bibitem[ALM{\etalchar{+}}98]{ALMSS}
Sanjeev Arora, Carsten Lund, Rajeev Motwani, Madhu Sudan, and Mario Szegedy.
\newblock Proof verification and the hardness of approximation problems.
\newblock {\em J. ACM}, 45(3):501--555, May 1998.

\bibitem[AMS06]{AlonMS06}
Noga Alon, Dana Moshkovitz, and Shmuel Safra.
\newblock Algorithmic construction of sets for $k$-restrictions.
\newblock {\em {ACM} Trans. Algorithms}, 2(2):153--177, 2006.

\bibitem[ANSW17]{AhmadianNSW17}
Sara Ahmadian, Ashkan Norouzi{-}Fard, Ola Svensson, and Justin Ward.
\newblock Better guarantees for $k$-means and euclidean $k$-median by
  primal-dual algorithms.
\newblock In {\em 58th {IEEE} Annual Symposium on Foundations of Computer
  Science, {FOCS} 2017, Berkeley, CA, USA, October 15-17, 2017}, pages 61--72,
  2017.

\bibitem[AS98]{AS}
Sanjeev Arora and Shmuel Safra.
\newblock Probabilistic checking of proofs: A new characterization of {NP}.
\newblock {\em J. ACM}, 45(1):70--122, January 1998.

\bibitem[AS03]{AroraS03}
Sanjeev Arora and Madhu Sudan.
\newblock Improved low-degree testing and its applications.
\newblock {\em Combinatorica}, 23(3):365--426, 2003.

\bibitem[AS18]{AggarwalS18}
Divesh Aggarwal and Noah Stephens{-}Davidowitz.
\newblock {(Gap/S)ETH} hardness of {SVP}.
\newblock In {\em Proceedings of the 50th Annual {ACM} {SIGACT} Symposium on
  Theory of Computing, {STOC} 2018, Los Angeles, CA, USA, June 25-29, 2018},
  pages 228--238, 2018.

\bibitem[BBE{\etalchar{+}}19]{evenset-merged}
Arnab Bhattacharyya, {\'{E}}douard Bonnet, L{\'{a}}szl{\'{o}} Egri, Suprovat
  Ghoshal, {Karthik {C. S.}}, Bingkai Lin, Pasin Manurangsi, and D{\'{a}}niel
  Marx.
\newblock Parameterized intractability of even set and shortest vector problem.
\newblock {\em Electronic Colloquium on Computational Complexity {(ECCC)}},
  26:115, 2019.

\bibitem[BELM18]{BELM}
{\'E}douard Bonnet, L{\'a}szl{\'o} Egri, Bingkai Lin, and D{\'a}niel Marx.
\newblock Fixed-parameter approximability of boolean {MinCSP}s.
\newblock {\em arXiv preprint arXiv:1601.04935}, 2018.

\bibitem[BGKM18]{BGKM18}
Arnab Bhattacharyya, Suprovat Ghoshal, {Karthik {C. S.}}, and Pasin Manurangsi.
\newblock Parameterized intractability of even set and shortest vector problem
  from {Gap-ETH}.
\newblock In {\em 45th International Colloquium on Automata, Languages, and
  Programming, {ICALP} 2018, July 9-13, 2018, Prague, Czech Republic}, pages
  17:1--17:15, 2018.

\bibitem[BPR{\etalchar{+}}17]{ByrkaPRST17}
Jaroslaw Byrka, Thomas Pensyl, Bartosz Rybicki, Aravind Srinivasan, and Khoa
  Trinh.
\newblock An improved approximation for $k$-median and positive correlation in
  budgeted optimization.
\newblock {\em {ACM} Trans. Algorithms}, 13(2):23:1--23:31, 2017.

\bibitem[CCK{\etalchar{+}}17]{CCKLM17}
Parinya Chalermsook, Marek Cygan, Guy Kortsarz, Bundit Laekhanukit, Pasin
  Manurangsi, Danupon Nanongkai, and Luca Trevisan.
\newblock From {Gap-ETH} to {FPT}-inapproximability: Clique, dominating set,
  and more.
\newblock In {\em 58th {IEEE} Annual Symposium on Foundations of Computer
  Science, {FOCS} 2017, Berkeley, CA, USA, October 15-17, 2017}, pages
  743--754, 2017.

\bibitem[CGK{\etalchar{+}}19]{CGKLL19}
Vincent Cohen{-}Addad, Anupam Gupta, Amit Kumar, Euiwoong Lee, and Jason Li.
\newblock Tight {FPT} approximations for $k$-median and $k$-means.
\newblock {\em CoRR}, abs/1904.12334, 2019.

\bibitem[CGTS02]{CharikarGTS02}
Moses Charikar, Sudipto Guha, {\'{E}}va Tardos, and David~B. Shmoys.
\newblock A constant-factor approximation algorithm for the $k$-median problem.
\newblock {\em J. Comput. Syst. Sci.}, 65(1):129--149, 2002.

\bibitem[Cha16]{Chan16}
Siu~On Chan.
\newblock Approximation resistance from pairwise-independent subgroups.
\newblock {\em J. {ACM}}, 63(3):27:1--27:32, 2016.

\bibitem[CL16]{ChenL16}
Yijia Chen and Bingkai Lin.
\newblock The constant inapproximability of the parameterized dominating set
  problem.
\newblock In {\em {IEEE} 57th Annual Symposium on Foundations of Computer
  Science, {FOCS} 2016, 9-11 October 2016, Hyatt Regency, New Brunswick, New
  Jersey, {USA}}, pages 505--514, 2016.

\bibitem[CW12]{ChengW12}
Qi~Cheng and Daqing Wan.
\newblock A deterministic reduction for the gap minimum distance problem.
\newblock {\em {IEEE} Trans. Information Theory}, 58(11):6935--6941, 2012.

\bibitem[DD19]{DiksteinD19}
Yotam Dikstein and Irit Dinur.
\newblock Agreement testing theorems on layered set systems.
\newblock {\em Electronic Colloquium on Computational Complexity {(ECCC)}},
  26:112, 2019.

\bibitem[DFH19]{DinurFH19}
Irit Dinur, Yuval Filmus, and Prahladh Harsha.
\newblock Analyzing boolean functions on the biased hypercube via
  higher-dimensional agreement tests: [extended abstract].
\newblock In {\em Proceedings of the Thirtieth Annual {ACM-SIAM} Symposium on
  Discrete Algorithms, {SODA} 2019, San Diego, California, USA, January 6-9,
  2019}, pages 2124--2133, 2019.

\bibitem[DFVW99]{DowneyFVW99}
Rodney~G. Downey, Michael~R. Fellows, Alexander Vardy, and Geoff Whittle.
\newblock The parametrized complexity of some fundamental problems in coding
  theory.
\newblock {\em {SIAM} J. Comput.}, 29(2):545--570, 1999.

\bibitem[DGKR05]{DinurGKR05}
Irit Dinur, Venkatesan Guruswami, Subhash Khot, and Oded Regev.
\newblock A new multilayered {PCP} and the hardness of hypergraph vertex cover.
\newblock {\em {SIAM} J. Comput.}, 34(5):1129--1146, 2005.

\bibitem[Din16]{D16}
Irit Dinur.
\newblock Mildly exponential reduction from gap {3SAT} to polynomial-gap
  label-cover.
\newblock {\em ECCC}, 23:128, 2016.

\bibitem[DK17]{DinurK17}
Irit Dinur and Tali Kaufman.
\newblock High dimensional expanders imply agreement expanders.
\newblock In {\em 58th {IEEE} Annual Symposium on Foundations of Computer
  Science, {FOCS} 2017, Berkeley, CA, USA, October 15-17, 2017}, pages
  974--985, 2017.

\bibitem[DM18a]{DinurM18}
Irit Dinur and Pasin Manurangsi.
\newblock {ETH}-hardness of approximating 2-{CSP}s and directed steiner
  network.
\newblock In {\em ITCS}, pages 36:1--36:20, 2018.

\bibitem[DM18b]{DinurM18-arxiv}
Irit Dinur and Pasin Manurangsi.
\newblock {ETH}-hardness of approximating 2-{CSP}s and directed steiner
  network.
\newblock {\em CoRR}, abs/1805.03867, 2018.

\bibitem[DMS03]{DumerMS03}
Ilya Dumer, Daniele Micciancio, and Madhu Sudan.
\newblock Hardness of approximating the minimum distance of a linear code.
\newblock {\em {IEEE} Trans. Information Theory}, 49(1):22--37, 2003.

\bibitem[DN17]{DinurN17}
Irit Dinur and Inbal~Livni Navon.
\newblock Exponentially small soundness for the direct product z-test.
\newblock In {\em 32nd Computational Complexity Conference, {CCC} 2017, July
  6-9, 2017, Riga, Latvia}, pages 29:1--29:50, 2017.

\bibitem[DR06]{DinurR06}
Irit Dinur and Omer Reingold.
\newblock Assignment testers: Towards a combinatorial proof of the {PCP}
  theorem.
\newblock {\em {SIAM} J. Comput.}, 36(4):975--1024, 2006.

\bibitem[DS14a]{DinurS14-parallel-rep}
Irit Dinur and David Steurer.
\newblock Analytical approach to parallel repetition.
\newblock In {\em Symposium on Theory of Computing, {STOC} 2014, New York, NY,
  USA, May 31 - June 03, 2014}, pages 624--633, 2014.

\bibitem[DS14b]{DinurS14}
Irit Dinur and David Steurer.
\newblock Direct product testing.
\newblock In {\em {IEEE} 29th Conference on Computational Complexity, {CCC}
  2014, Vancouver, BC, Canada, June 11-13, 2014}, pages 188--196, 2014.

\bibitem[Fei98]{Feige98}
Uriel Feige.
\newblock A threshold of $\ln n$ for approximating set cover.
\newblock {\em J. {ACM}}, 45(4):634--652, 1998.

\bibitem[FGL{\etalchar{+}}91]{FGLSS}
Uriel Feige, Shafi Goldwasser, L{\'{a}}szl{\'{o}} Lov{\'{a}}sz, Shmuel Safra,
  and Mario Szegedy.
\newblock Approximating clique is almost {NP}-complete (preliminary version).
\newblock In {\em 32nd Annual Symposium on Foundations of Computer Science, San
  Juan, Puerto Rico, 1-4 October 1991}, pages 2--12, 1991.

\bibitem[GK99]{GuhaK99}
Sudipto Guha and Samir Khuller.
\newblock Greedy strikes back: Improved facility location algorithms.
\newblock {\em J. Algorithms}, 31(1):228--248, 1999.

\bibitem[GK18]{GoldenbergS18}
Elazar Goldenberg and {Karthik {C. S.}}
\newblock Towards a general direct product testing theorem.
\newblock In {\em 38th {IARCS} Annual Conference on Foundations of Software
  Technology and Theoretical Computer Science, {FSTTCS} 2018, December 11-13,
  2018, Ahmedabad, India}, pages 11:1--11:17, 2018.

\bibitem[GL19]{GL19}
Venkatesan Guruswami and Patrick Lin.
\newblock Parameterized inapproximability of exact cover and nearest codeword.
\newblock {\em CoRR}, abs/1905.06503, 2019.

\bibitem[GS00]{GoldreichS00}
Oded Goldreich and Shmuel Safra.
\newblock A combinatorial consistency lemma with application to proving the
  {PCP} theorem.
\newblock {\em {SIAM} J. Comput.}, 29(4):1132--1154, 2000.

\bibitem[GT08]{GT08}
Anupam Gupta and Kanat Tangwongsan.
\newblock Simpler analyses of local search algorithms for facility location.
\newblock {\em CoRR}, abs/0809.2554, 2008.

\bibitem[H{\aa}s96]{Hastad96}
Johan H{\aa}stad.
\newblock Clique is hard to approximate within $n^{1 - \epsilon}$.
\newblock In {\em 37th Annual Symposium on Foundations of Computer Science,
  {FOCS} '96, Burlington, Vermont, USA, 14-16 October, 1996}, pages 627--636,
  1996.

\bibitem[H{\aa}s01]{Hastad01}
Johan H{\aa}stad.
\newblock Some optimal inapproximability results.
\newblock {\em J. {ACM}}, 48(4):798--859, 2001.

\bibitem[Hoc97]{Hochba:1997:AAN:261342.571216}
Approximation algorithms for {NP}-hard problems.
\newblock {\em SIGACT News}, 28(2):40--52, June 1997.

\bibitem[HR12]{HavivR12}
Ishay Haviv and Oded Regev.
\newblock Tensor-based hardness of the shortest vector problem to within almost
  polynomial factors.
\newblock {\em Theory of Computing}, 8(1):513--531, 2012.

\bibitem[IKW12]{ImpagliazzoKW12}
Russell Impagliazzo, Valentine Kabanets, and Avi Wigderson.
\newblock New direct-product testers and 2-query {PCP}s.
\newblock {\em {SIAM} J. Comput.}, 41(6):1722--1768, 2012.

\bibitem[IP01]{ImpagliazzoP01}
Russell Impagliazzo and Ramamohan Paturi.
\newblock On the complexity of k-{SAT}.
\newblock {\em J. Comput. Syst. Sci.}, 62(2):367--375, 2001.

\bibitem[IPZ01]{ImpagliazzoPZ01}
Russell Impagliazzo, Ramamohan Paturi, and Francis Zane.
\newblock Which problems have strongly exponential complexity?
\newblock {\em J. Comput. Syst. Sci.}, 63(4):512--530, 2001.

\bibitem[JMS02]{JainMS02}
Kamal Jain, Mohammad Mahdian, and Amin Saberi.
\newblock A new greedy approach for facility location problems.
\newblock In {\em Proceedings on 34th Annual {ACM} Symposium on Theory of
  Computing, May 19-21, 2002, Montr{\'{e}}al, Qu{\'{e}}bec, Canada}, pages
  731--740, 2002.

\bibitem[Joh74]{Johnson74a}
David~S. Johnson.
\newblock Approximation algorithms for combinatorial problems.
\newblock {\em J. Comput. Syst. Sci.}, 9(3):256--278, 1974.

\bibitem[JV01]{JainV01}
Kamal Jain and Vijay~V. Vazirani.
\newblock Approximation algorithms for metric facility location and $k$-median
  problems using the primal-dual schema and lagrangian relaxation.
\newblock {\em J. {ACM}}, 48(2):274--296, 2001.

\bibitem[Kho02a]{Khot02-b}
Subhash Khot.
\newblock Hardness results for coloring 3 -colorable 3 -uniform hypergraphs.
\newblock In {\em 43rd Symposium on Foundations of Computer Science {(FOCS}
  2002), 16-19 November 2002, Vancouver, BC, Canada, Proceedings}, pages
  23--32, 2002.

\bibitem[Kho02b]{Khot02}
Subhash Khot.
\newblock On the power of unique 2-prover 1-round games.
\newblock In {\em Proceedings of the 17th Annual {IEEE} Conference on
  Computational Complexity, Montr{\'{e}}al, Qu{\'{e}}bec, Canada, May 21-24,
  2002}, page~25, 2002.

\bibitem[Kho05]{Khot05}
Subhash Khot.
\newblock Hardness of approximating the shortest vector problem in lattices.
\newblock {\em J. {ACM}}, 52(5):789--808, 2005.

\bibitem[KLM18]{KLM18}
{Karthik {C. S.}}, Bundit Laekhanukit, and Pasin Manurangsi.
\newblock On the parameterized complexity of approximating dominating set.
\newblock In {\em Proceedings of the 50th Annual {ACM} {SIGACT} Symposium on
  Theory of Computing, {STOC} 2018, Los Angeles, CA, USA, June 25-29, 2018},
  pages 1283--1296, 2018.

\bibitem[KMN{\etalchar{+}}04]{KanungoMNPSW04}
Tapas Kanungo, David~M. Mount, Nathan~S. Netanyahu, Christine~D. Piatko, Ruth
  Silverman, and Angela~Y. Wu.
\newblock A local search approximation algorithm for $k$-means clustering.
\newblock {\em Comput. Geom.}, 28(2-3):89--112, 2004.

\bibitem[Lin19]{Lin19}
Bingkai Lin.
\newblock A simple gap-producing reduction for the parameterized set cover
  problem.
\newblock {\em CoRR}, abs/1902.03702, 2019.

\bibitem[Lov75]{Lovasz75}
L.~Lov{\'{a}}sz.
\newblock On the ratio of optimal integral and fractional covers.
\newblock {\em Discrete Mathematics}, 13(4):383--390, 1975.

\bibitem[LRSZ17]{LRSZ17}
Daniel Lokshtanov, M.~S. Ramanujan, Saket Saurabh, and Meirav Zehavi.
\newblock Parameterized complexity and approximability of directed odd cycle
  transversal.
\newblock {\em CoRR}, abs/1704.04249, 2017.

\bibitem[LS16]{LiS16}
Shi Li and Ola Svensson.
\newblock Approximating $k$-median via pseudo-approximation.
\newblock {\em {SIAM} J. Comput.}, 45(2):530--547, 2016.

\bibitem[LY94]{LundY94}
Carsten Lund and Mihalis Yannakakis.
\newblock On the hardness of approximating minimization problems.
\newblock {\em J. {ACM}}, 41(5):960--981, 1994.

\bibitem[Man19]{M-phd-thesis}
Pasin Manurangsi.
\newblock {\em Approximation and Hardness: Beyond {P} and {NP}}.
\newblock PhD thesis, EECS Department, University of California, Berkeley, May
  2019.

\bibitem[Mar10]{Marx10}
D{\'{a}}niel Marx.
\newblock Can you beat treewidth?
\newblock {\em Theory of Computing}, 6(1):85--112, 2010.

\bibitem[Mic00]{Micciancio00}
Daniele Micciancio.
\newblock The shortest vector in a lattice is hard to approximate to within
  some constant.
\newblock {\em {SIAM} J. Comput.}, 30(6):2008--2035, 2000.

\bibitem[Mic14]{Micciancio14}
Daniele Micciancio.
\newblock Locally dense codes.
\newblock In {\em {IEEE} 29th Conference on Computational Complexity, {CCC}
  2014, Vancouver, BC, Canada, June 11-13, 2014}, pages 90--97, 2014.

\bibitem[Mos15]{Moshkovitz15}
Dana Moshkovitz.
\newblock The projection games conjecture and the {NP}-hardness of $\ln
  n$-approximating set-cover.
\newblock {\em Theory of Computing}, 11:221--235, 2015.

\bibitem[MR17]{MR17-icalp}
Pasin Manurangsi and Prasad Raghavendra.
\newblock A birthday repetition theorem and complexity of approximating dense
  {CSP}s.
\newblock In {\em ICALP}, pages 78:1--78:15, 2017.

\bibitem[PY91]{PapadimitriouY91}
Christos~H. Papadimitriou and Mihalis Yannakakis.
\newblock Optimization, approximation, and complexity classes.
\newblock {\em J. Comput. Syst. Sci.}, 43(3):425--440, 1991.

\bibitem[Raz98]{Raz98}
Ran Raz.
\newblock A parallel repetition theorem.
\newblock {\em {SIAM} J. Comput.}, 27(3):763--803, 1998.

\bibitem[RS97]{RazS97}
Ran Raz and Shmuel Safra.
\newblock A sub-constant error-probability low-degree test, and a sub-constant
  error-probability {PCP} characterization of {NP}.
\newblock In {\em Proceedings of the Twenty-Ninth Annual {ACM} Symposium on the
  Theory of Computing, El Paso, Texas, USA, May 4-6, 1997}, pages 475--484,
  1997.

\bibitem[SV19]{SV19}
Noah Stephens{-}Davidowitz and Vinod Vaikuntanathan.
\newblock {SETH}-hardness of coding problems.
\newblock In {\em FOCS}, 2019.
\newblock To appear.

\end{thebibliography}

\appendix

\section{$\wval$ vs $\val$ for Label Cover: Proof of Theorem~\ref{thm:main-label-cover}}
\label{app:weak-val}

In this section, we provide a simple proof that Theorem~\ref{thm:main-label-cover-weak-agr} implies Theorem~\ref{thm:main-label-cover}. The main observation here is that, if $\cL$ is biregular with left degree $t$, then $\wval(\cL) < \delta$ implies $\val(\cL) < \delta + \frac{1}{t}$.

\begin{proof}[Proof of Theorem~\ref{thm:main-label-cover} from Theorem~\ref{thm:main-label-cover-weak-agr}]
Suppose contrapositively that there exists a $T(k) \cdot N^{o(k)}$-time algorithm that can solve $\nu$-\textsc{Gap-Label-Cover} for some $\nu > 0$. We claim that this algorithm can also distinguish the following two cases in Theorem~\ref{thm:main-label-cover-weak-agr}, with parameters $t = \frac{2}{\nu}$ and $\delta = \frac{\nu}{2}$. The completeness of the two theorems are exactly the same; hence, it suffices to argue that, if $\wval(\cL) < \delta$, then $\val(\cL) < \nu$. To see that this is the case, consider any labeling $\sigma = (\sigma_w)_{w \in U \cup V}$ of $\cL$. Let $\sigma^L = (\sigma_u)_{u \in U}$ be the left labeling induced by $\sigma$. Observe that if $\sigma^L$ does not weakly agree on a right vertex $u \in U$, then at most one of the edges adjacent to it can be satisfied by $\sigma$. Since $\wval(\sigma^L) < \delta$, we must have $\val(\sigma) < \delta + (1 - \delta)(1/t) < \nu$. 

Hence, $\val(\cL) < \delta$ as desired, and, from Theorem~\ref{thm:main-label-cover-weak-agr}, we violate Gap-ETH.
\end{proof}

\section{Proof of Lemma~\ref{lem:majority}}
\label{app:majority}

Here, we prove Lemma~\ref{lem:majority}. The proof is almost quoted from~\cite{DinurM18}, except that, in~\eqref{ineq:maj-upper}, we can bound $\disagr(f_{S_1}, f_{S_2})$ by $\rho n$ because we assume that $|S_1 \cap S_2| \leq \rho n$. In~\cite{DinurM18}, no such assumption is made and hence the upper bound there is only the (trivial) bound of $n$.

\begin{proof}[Proof of Lemma~\ref{lem:majority}]
We have
\begin{align}
\E_{S_1, S_2 \in \cS'}[\disa(f_{S_1}, f_{S_2})] &\leqs \Pr_{S_1, S_2 \in \cS'}[\disagr(f_{S_1}, f_{S_2}) > \zeta n] \cdot (\rho n) + \Pr_{S_1, S_2 \in \cS'}[\disagr(f_{S_1}, f_{S_2}) \leq \zeta n] \cdot (\zeta n) \label{ineq:maj-upper} \\
&\leqs (\kappa \rho + \zeta) n.
\end{align}

We can then lower bound the expression on the left hand side as follows.
\begin{align}
\EX_{S_1, S_2 \in \cS'}[\disa(f_{S_1}, f_{S_2})] &= \sum_{x \in [n]} \Pr_{S_1, S_2 \in \cS'}\left[x \in S_1 \wedge x \in S_2 \wedge f_{S_1}(x) \ne f_{S_2}(x)\right] \nonumber \\
&\geqs \sum_{x \in [n]} \Pr_{S_1, S_2 \in \cS'}[x \in S_1 \wedge x \in S_2 \wedge f_{S_1}(x) \ne g(x) \wedge f_{S_2}(x) = g(x)] \nonumber \\
&= \sum_{x \in [n]} \Pr_{S_1 \in \cS'}[x \in S_1 \wedge f_{S_1}(x) \ne g(x)]\Pr_{S_2 \in \cS'}[x \in S_2 \wedge f_{S_2}(x) = g(x)] \nonumber \\
(\text{Since } g(x) = \maj_{S \in \cS' \atop x \in S} (f_{S}(x))) &\geqs \sum_{x \in [n]} \Pr_{S_1 \in \cS'}[x \in S_1 \wedge f_{S_1}(x) \ne g(x)]\Pr_{S_2 \in \cS'}[x \in S_2 \wedge f_{S_2}(x) \ne g(x)] \nonumber \\
&= \sum_{x \in [n]} \left(\Pr_{S \in \cS'}[x \in S \wedge f_S(x) \ne g(x)]\right)^2 \nonumber \\
(\text{Power Mean Inequality}) &\geqs \frac{1}{n} \left(\sum_{x \in [n]}\Pr_{S \in \cS'}[x \in S \wedge f_S(x) \ne g(x)]\right)^2 \nonumber \\
&= \frac{1}{n} \left(\EX_{S \in \cS'}\left[\disa(g, f_S)\right]\right)^2. \label{ineq:maj-lower}
\end{align}
Combining (\ref{ineq:maj-upper}) and (\ref{ineq:maj-lower}) gives the desired bound.
\end{proof}

\end{document}